\documentclass[11pt]{article}
\usepackage{fullpage,amsfonts,amsmath,amsthm,amssymb,mathtools,mathrsfs,graphicx,upgreek,tablefootnote}
\usepackage[margin=1in,marginparwidth=0.75in]{geometry}

\usepackage{algorithm}
\usepackage[noend]{algpseudocode}

\usepackage{epsfig}
\usepackage{bm}
\usepackage{comment}
\usepackage{hyperref}
\usepackage{cleveref}
\crefname{subsection}{subsection}{subsections}

\usepackage{enumitem}

\numberwithin{equation}{section}

\newtheorem{theorem}{Theorem}[section]
\newtheorem{proposition}[theorem]{Proposition}

\newtheorem{lemma}[theorem]{Lemma}

\theoremstyle{definition}

\newcommand{\Bin}{\mathrm{Bin}}

\newtheorem{remark}[theorem]{Remark}

\usepackage{bbm}

\usepackage{algpseudocode}
\usepackage{datetime}

\usepackage[english]{babel}
\usepackage[autostyle, english = american]{csquotes}
\MakeOuterQuote{"}

\newcommand{\scr}{\mathcal}
\newcommand{\mb}{\mathbb}

\newcommand{\E}{\mathbb E}
\newcommand{\eps}{\varepsilon}

\newcommand{\hs}{\hat{s}}
\newcommand{\citet}{\cite}
\newcommand{\citep}{\cite}
\newcommand{\bE}{\mathbb{E}}
\newcommand{\bP}{\mathbb{P}}

\newcommand{\ra}{\rightarrow}
\newcommand{\sfe}{\mathsf{e}}

\newcommand{\LB}{\scr{L}}
\newcommand{\UB}{\scr{U}}

\newcommand{\Ber}{\mathrm{Ber}}

\newcommand{\hS}{\hat{S}}

\usepackage[usenames,dvipsnames]{xcolor}
\usepackage{tikz}
\usetikzlibrary{arrows.meta}

\begin{document}
\title{
Random-order Contention Resolution via Continuous Induction: Tightness for Bipartite Matching under Vertex Arrivals
}

\author{
Calum MacRury
\thanks{Graduate School of Business, Columbia University.
\texttt{cm4379@columbia.edu}}
\and
Will Ma
\thanks{Graduate School of Business and Data Science Institute, Columbia University.
\texttt{wm2428@gsb.columbia.edu}
}
}

\date{}
\maketitle

\begin{abstract}
We introduce a new approach for designing Random-order Contention Resolution Schemes (RCRS) via exact solution in continuous time. Given a function $c(y):[0,1] \ra [0,1]$,
we show how to select each element which arrives at time $y \in [0,1]$ with probability \textit{exactly} $c(y)$.
We provide a rigorous algorithmic framework for achieving this, which discretizes the time interval and also needs to sample its past execution to ensure these exact selection probabilities.
We showcase
our framework in the context of online contention resolution schemes for matching with random-order vertex arrivals. For bipartite graphs with two-sided arrivals, we design a $(1+\sfe^{-2})/2 \approx 0.567$-selectable RCRS, which we also show to be \textit{tight}.
Next, we show that the presence of short odd-length cycles is the only barrier to attaining a (tight) $(1+\sfe^{-2})/2$-selectable RCRS on general graphs. By generalizing our bipartite RCRS, we design an RCRS for graphs with odd-length girth $g$ which is $(1+\sfe^{-2})/2$-selectable as $g \rightarrow \infty$. This convergence happens very rapidly: for triangle-free
graphs (i.e., $g \ge 5$), we attain a  $121/240 + 7/16 e^2 \approx 0.563$-selectable RCRS.
Finally, for general graphs we improve on the $8/15 \approx 0.533$-selectable RCRS of \citet{Fu2021} and design an RCRS which is at least $0.535$-selectable.
Due to the reduction of \citet{Ezra_2020}, our bounds yield a $0.535$-competitive (respectively, $(1+\sfe^{-2})/2$-competitive) algorithm for prophet secretary matching
on general (respectively, bipartite) graphs under vertex arrivals.
\end{abstract}

\section{Introduction} 

Prophet inequalities compare the performance of online vs.\ offline selection algorithms in an abstract, parsimonious setting.
In this setting, there is a universe of elements, each with a valuation that is initially unknown but drawn from a known and independent distribution.
The elements arrive one by one, revealing their valuation upon arrival, at which point an online algorithm must irrevocably decide whether to accept or reject the element.
Feasibility constraints prevent the algorithm from accepting all elements, and the online algorithm's objective is to maximize the expected total valuation of accepted elements.

This paper focuses on feasibility constraints defined by matchings in a graph $G=(V,E)$.  The elements are the edges $e\in E$, each with a random valuation or "utility" $U_e$.  A subset of edges $M\subseteq E$ forms a \textit{matching} if no two of them share a vertex.  In this case, an arriving edge $e$ is feasible to accept as long as it does not share a vertex with a previously-accepted edge; however, the online algorithm can still optionally reject $e$ if the revealed $U_e$ is small compared to the opportunity cost of "depleting" the two vertices in $e$.  We note that exact opportunity costs are generally intractable to compute since the state space is exponential in $|V|$, hence it is a priori unclear how to even design an online algorithm to approximate the dynamic programming solution.

In prophet inequalities, the online algorithm is compared to a stronger benchmark of the "prophet", who sees all realized valuations in advance and can accept the matching $M$ that maximizes $\sum_{e\in M} U_e$.  Let $\mathcal{F}$ denote the family of all possible matchings in graph $G$.  A prophet inequality is then a guarantee that an online algorithm's expected total valuation is at least $\alpha\cdot \bE[\max_{M\in\mathcal{F}}\sum_{e\in M} U_e]$, for some constant $\alpha\in[0,1]$, that holds for \textit{all} graphs $G$ (possibly restricted to some class, e.g.\ bipartite) and independent distributions for the utilities $U_e$.
Since the prophet provides an upper bound on the expected total valuation of any dynamic programming solution, a prophet inequality guarantee of $\alpha$ also implies an $\alpha$-approximation algorithm for this online stochastic optimization problem on graphs.

\textbf{Reduction to Online Contention Resolution Schemes (OCRS)}.
Prophet inequalities are often derived via OCRS's, which are algorithms that operate in the following alternate setting.
Elements $e\in E$ now have a binary random state, that is \textit{active} independently with a known probability $x_e$, and inactive otherwise.
"Active" should be interpreted as realizing a valuation "worthy of acceptance", and only active elements can be accepted.
The elements arrive one by one, revealing their state upon arrival, at which point active elements can be irrevocably accepted or rejected by the OCRS (inactive elements must be rejected), subject to the same feasibility constraints.  Again we focus on the feasibility constraint that the final subset of accepted elements forms a matching in the graph.
The objective is to accept every edge $e$ with probability (w.p.) at least $\alpha$ conditional on it being active, for a constant $\alpha\in[0,1]$ as large as possible.
An OCRS is said to be \textit{$\alpha$-selectable} (equivalently, we say it has \textit{selectability} $\alpha$) if it can achieve this for all graphs $G$ and vectors $(x_e)_{e\in E}$ satisfying
\begin{align} \label{eqn:introMatchingPolytope}
\sum_{e\ni v}x_e &\le1 &\forall v\in V.
\end{align}

The motivation for~\eqref{eqn:introMatchingPolytope} is that the expected number of active edges around any vertex $v$ should not exceed 1, since at most one of them can be accepted.
It is well-known that an $\alpha$-selectable OCRS can be converted into an online algorithm for the original problem whose expected total valuation is at least $\alpha\cdot \bE[\max_{M\in\mathcal{F}}\sum_{e\in M} U_e]$, providing a prophet inequality guarantee of $\alpha$ for matchings in graphs.  We defer a proof of this fact to \citet{feldman2021online, Ezra_2020} and hereafter focus on OCRS's.

\textbf{Vertex-batched arrivals}. \citet{Ezra_2020} introduce a generalization of prophet inequality where elements arrive in batches; the online algorithm sees the realized valuations of all elements in a batch before committing to any decisions and moving to the next batch.  The original online algorithms problem is recovered if each batch consists of a single element; the prophet's problem is recovered if there is a single batch consisting of all elements.  For graphs, a natural intermediate granularity of batching is vertex arrivals: the vertices arrive one by one, and all edges between the arriving vertex and previously-arrived vertices are revealed at the same time.  In the analogous OCRS, the binary random states are revealed in the same batches, with the additional constraint that there is \textit{at most one active edge} per batch. This additional constraint comes from the fact that at most one edge per batch can be matched, under vertex arrivals. We defer the details of the reduction to \citet{Ezra_2020}, where the authors prove that an $\alpha$-selectable OCRS with this additional constraint implies a prophet inequality guarantee of $\alpha$ under vertex-batched arrivals.

\textbf{Random-order arrivals}. Aside from batching, the guarantees $\alpha$ that are achievable for prophet inequalities and OCRS's depend heavily on the order in which elements arrive.
This paper assumes that elements arrive in a uniformly random order, in which case an OCRS can be called a \textit{Random-order Contention Resolution Scheme (RCRS)}.
We adopt the standard re-formulation of random-order arrivals where each element $e$ arrives at a random time $Y_e$, that is drawn independently and uniformly from the continuous time horizon [0,1].

\subsection{New Framework: Exact Selection in Continuous Time} \label{sec:introContInd}
Suppose that element $e\in E$ arrives at time $Y_e\in[0,1]$ and is active.
The RCRS only has a decision if $e$ is feasible to accept.
A common approach to making this decision is to accept $e$ following an independent Bernoulli random bit with probability $f_e(Y_e)$, which can depend on both the element $e$ and its realized arrival time $Y_e$ (e.g.~\citet{Lee2018} use the function $f_e(y)=\sfe^{-yx_e}$, recalling that $x_e$ is the activeness probability of $e$).
Integrating over arrival time, the probability of accepting any element $e$ conditional on it being active is then
\begin{align} \label{eqn:typicalSelGuarantee}
\int_0^1 f_e(y) \bP[e\text{ is feasible}\mid Y_e=y]dy,
\end{align}
and the RCRS is $\alpha$-selectable where $\alpha$ is the minimum value of~\eqref{eqn:typicalSelGuarantee} over elements $e$.
We note that $f_e$ does not necessarily need to depend on arrival time $y$ (e.g.~\citet{brubach2021offline} use the function $f_e(y)=\sfe^{-x_e/2}$, among others); however in existing RCRS's, $f_e$ has always been given as an explicit function, sometimes called an attenuation function.

Our framework differs by defining $f_e$ implicitly, with $f_e(y):=c(y)/\bP[e\text{ is feasible}\mid Y_e=y]$.  Substituting into~\eqref{eqn:typicalSelGuarantee}, the probability of accepting any element $e$ conditional on it being active is then $\int_0^1 c(y) dy$, and the RCRS is $(\int_0^1 c(y) dy)$-selectable.  We call $c:[0,1]\to[0,1]$ a \textit{selection function}.  Note that our RCRS accepts all elements with the \textit{same, exact} probability $c(y)$, conditional on it arriving at time $y$ and being active.

Our RCRS is only well-defined if the probabilities $f_e(y)$ are at most 1, which requires
\begin{align} \label{eqn:goal}
c(y) &\le \bP[e\text{ is feasible}\mid Y_e=y] &\forall e\in E, y\in[0,1].
\end{align}
We design selection functions $c$ that satisfy this property, which we establish inductively for any $e\in E$ and $t \in[0,1]$ by assuming the property is satisfied for elements arriving at all previous times $y\in[0,t)$.  In this sense, we are applying a "continuous induction" argument over real interval $[0,1]$.  Moreover, the benefit of exact selection is that it is easier to lower-bound $\bP[e\text{ is feasible}\mid Y_e=t]$ in~\eqref{eqn:goal} using the continuous induction hypothesis, since $e$ is feasible at time $t$ as long as conflicting elements (in this case, edges $e'$ that share a vertex with $e$) were not accepted before $t$, which we know happens with exact probabilities $c(y)$ for $y\in[0,t)$.

The idea of exact selection has been previously used by \citet{Ezra_2020} in the design of OCRS, where the arrival order of past elements is fixed, and the selection function is a constant $c$ instead of dependent on the arrival time $y$.  We show this idea to be powerful in the design of RCRS, where we can now integrate over random arrival times.
It is especially well-suited for designing RCRS's under vertex-batched arrivals, but first, we warm up by demonstrating the power of our framework in non-batched settings, recovering some known results and also easily establishing a new $1/2$-selectable RCRS on trees (which, as we will explain, is not easily established from previous approaches).

\textbf{Rank-1 matroids: $1-1/\sfe\approx0.632$-selectable RCRS \citep{Lee2018}}. Rank-$1$ matroids describe the most basic feasibility constraint where at most one element can be accepted.  We represent this using a star graph with all edges $e\in E$ containing a central vertex $v$, in which case we know $\sum_{e\in E} x_e\le1$ from~\eqref{eqn:introMatchingPolytope}. To establish the desired~\eqref{eqn:goal}, note that $e$ is feasible if and only if no edge was accepted before time $t$.  We claim that
\begin{align} \label{eqn:dropConditioning}
\bP[\text{no edge was accepted before $t$}\mid Y_e=t]\ge \bP[\text{no edge was accepted before $t$}],
\end{align}
which is subtle but important.  Indeed, conditioning on $Y_e=t$ can only increase the probability that no edge was accepted before $t$, since it means $e$ cannot be accepted before time $t$, also noting that conditioning on $Y_e$ does not affect the independent probabilities with which our RCRS attempts to accept other elements.  Only after dropping the conditioning can we derive
\begin{align} \label{eqn:rank1}
\bP[\text{no edge was accepted before $t$}] = 1-\sum_{e'} x_{e'} \int_0^t c(y) dy,
\end{align}
where we have used the induction hypothesis that the (unconditional) probability of accepting any element $e'$ conditional on it arriving at time $y$ and being active (which occurs w.p.~$x_{e'}$) is $c(y)$.

We now analytically derive the best selection function.  In order for~\eqref{eqn:dropConditioning} and~\eqref{eqn:rank1} combined to establish~\eqref{eqn:goal}, we require $c(t)\le1-\int_0^tc(y)dy$ for all $t\in[0,1]$, recalling that $\sum_{e'} x_{e'}\le1$.  To maximize the guarantee $\alpha=\int_0^1 c(y)dy$ while satisfying this requirement, we solve the integral equation $c(t)=1-\int_0^t c(y)dy$ with boundary condition $c(0)\le1$, leading to the well-known selection function $c(y)=\sfe^{-y}$ and guarantee $\alpha=1-1/\sfe$.
Interestingly, if $\sum_{e\in E}x_e=1$, then our framework actually recovers the $(1-1/\sfe)$-selectable RCRS of \citet[Apx.~A.2]{Lee2018}, because it can be checked that $\bP[e\text{ is feasible}\mid Y_e=y]=\sfe^{-y(1-x_e)}$ and hence our $f_e(y)=\sfe^{-y}/\sfe^{-y(1-x_e)}=\sfe^{-yx_e}$. We verify this in \Cref{sec:rank_one}.

\textbf{General graphs: $(1-\sfe^{-2})/2\approx0.432$-selectable RCRS \citep{brubach2021offline}}. We extend the same framework from star graphs to general graphs.
An edge $e=(u,v)$ arriving at time $t$ is now feasible if and only if neither vertex $u$ nor $v$ has already been matched.  We can drop the conditioning on $Y_e=t$ as in~\eqref{eqn:dropConditioning}, following the same argument.
The analogue of~\eqref{eqn:rank1} is now
\begin{align} \label{eqn:432}
\bP[\text{neither $u$ nor $v$ matched before $t$}] \ge 1-\Big(\sum_{e'\ni u} x_{e'} + \sum_{e'\ni v} x_{e'}\Big)\int_0^t c(y) dy,
\end{align}
using the same induction hypothesis. Unlike~\eqref{eqn:rank1}, we now have an inequality in~\eqref{eqn:432} from the union bound, because the "bad" events of $u$ and $v$ being matched may not be mutually exclusive.

Applying~\eqref{eqn:introMatchingPolytope} for vertices $u$ and $v$, requirement~\eqref{eqn:goal} is satisfied as long as $c(t)\le1-2\int_0^t c(y)dy$ for all $t\in[0,1]$.  The integral equation is now $c(t)= 1-2\int_0^t c(y)dy$, leading to selection function $c(y)=\sfe^{-2y}$ and a guarantee of $\alpha=\int_0^1\sfe^{-2y}dy=(1-\sfe^{-2})/2$.

\textbf{Trees: $1/2$-selectable RCRS [New]}.
When the graph is a tree, a selection function with a better guarantee is possible.
This is because the outcomes (determined by activeness, arrival times, and the independent Bernoulli bits drawn by our RCRS) on one side of edge $e=(u,v)$ before time $t$ are now \textit{independent} from outcomes on the other side, \textit{conditional} on $e$ arriving at time $t$, i.e.
\begin{align}
&\bP[\text{neither $u$ nor $v$ matched before $t$}|Y_e=t] \nonumber
\\ &=\bP[\text{$v$ not matched before $t$}|Y_e=t]\bP[\text{$u$ not matched before $t$}|Y_e=t]. \label{eqn:treeIntro}
\end{align}
From here we drop the conditioning on $Y_e=t$ and then apply the induction hypothesis, like before:
\begin{align*}
\eqref{eqn:treeIntro} &\ge\bP[\text{$v$ not matched before $t$}]\bP[\text{$u$ not matched before $t$}]
\\ &=\Big(1-\sum_{e'\ni u} x_{e'}\int_0^t c(y) dy)\Big)\Big(1-\sum_{e'\ni v} x_{e'}\int_0^t c(y) dy\Big).
\end{align*}
Note that this sequence of arguments can be applied as long as the bad events of $u$ and $v$ being matched before time $t$, conditional on $Y_e=t$, are not negatively correlated.
The resulting integral equation is $c(t) = (1 - \int_{0}^{t} c(y) dy)^2$, and
its solution is $c(y)=1/(1+y)^2$, leading to a guarantee of $\alpha=\int_0^1 (1+y)^{-2} dy=1/2$. We remark that \cite{macrury2023random} establish an upper bound of 1/2 on the guarantee of RCRS that can be achieved for bipartite graphs (a superclass of trees).

This result does not appear to easily follow from modifying existing approaches to exploit independence.  Indeed, both the $(1-\sfe^{-2})/2$-selectable RCRS of \citet{brubach2021offline} and the alternate proof of a $(1-\sfe^{-2})/2$-selectable RCRS from \citet{pollner2022improved} (based on the attenuation function of \citet{Lee2018}) require defining explicit functions $f_e$,
and in fact are exactly\footnote{This can be seen from a graph $G=(V,E)$ with $V=\{u_i,v_i\}_{i=0}^n$ and $E=\{(u_0,v_0)\}\cup\{(u_0,u_i),(v_0,v_i)\}_{i=1}^n$. The activeness probability $x_e$ of every $e\in E$ is $1/(n+1)$.  It can be computed that both of their RCRS's select the middle edge $(u_0,v_0)$ with probability \textit{exactly} $(1-\sfe^{-2})/2$ as $n\to\infty$.} $(1-\sfe^{-2})/2$-selectable on trees. By contrast, in our framework it is clear how the integral equation changes under independence, and solving it directly yields a new algorithm and guarantee.

\textbf{Note on "continuous induction", and our discretization/sampling framework.}
We have thus far argued that our algorithm is well-defined at each time $t$ assuming it is well-defined at all previous times $y<t$, and hence well-defined for all $t\in[0,1]$ via "continuous induction". While there are formalisms of continuous induction on $[0,1]$ (see \cite{continuos_induction_2007}), they require verifying an additional property: one must show that if the algorithm is well-defined at $t \in [0,1]$, then it is also well-defined at
time $t + \delta_t$ for some $\delta_t > 0$. The challenge in verifying such a property is that one doesn't get to assume the algorithm is well-defined for $y \in [t, t+ \delta_t)$, and so the edges need not be selected with probability exactly $c(y)$.

We thus abandon this approach, and instead discretize our algorithm to make decisions in $T$ phases.  Since we also want our algorithms to be polyonimal-time, we must estimate the probabilities $\bP[e\text{ is feasible}\mid Y_e=y]$ (which affect the random bits used by our algorithm) using sampling.
We provide a rigorous framework for doing this discretization combined with sampling, in \Cref{sec:recursive_rcrs}.  We believe that this framework, which reduces RCRS to finding selection functions $c(y)$ that satisfy~\eqref{eqn:goal}, is a broader contribution of our work.  We now show that our framework yields particularly good results in the vertex-batched arrival setting.

\subsection{Results for Vertex-batched Arrivals} \label{sec:results}

The \textit{odd-girth} $g$ of a graph $G$ is defined as the length of the shortest odd cycle in $G$ \footnote{Technically, we can remove the edges $e\in E$ with $x_e=0$ from the edge set of $G$.}.  A bipartite graph, which contains no odd cycles, is said to have odd-girth $g=\infty$.
For each odd integer $g \ge 3$, we define a target selectability $\alpha_g$, where $\alpha_{\infty} := (1+\sfe^{-2})/2 \approx 0.5676$ (see \eqref{eqn:alpha_definition} in \Cref{sec:technical_overview} for an explicit definition of $\alpha_g$). The sequence $(\alpha_g)_g$ is increasing and
converges to $\alpha_{\infty}$ as $g \rightarrow \infty$. In particular, $\alpha_3 = 5/12 + 1/4 \sfe^2 \ge 0.4505$, $\alpha_5 = 121/240 + 7/16 \sfe^2 \ge 0.5633$,
and $\alpha_7 = 10121/20160 + 31/64 \sfe^2 \ge 0.5675$.

\begin{theorem} \label{thm:recursive_rcrs}
Fix $0 < \eps < 1$, and let $G=(V,E)$ be a graph with odd-girth $g$. Then an RCRS with runtime polynomial in the input $G,(x_e)_{e\in E}$ and $1/\eps$ is
$(1-\eps)^2 \alpha_g$-selectable \footnote{The multiplicative $(1-\eps)^2$ error term is to simplify the computations in the proof of \Cref{thm:recursive_rcrs}.  The same statement holds if we want the RCRS to be $(1-\eps)\alpha_g$-selectable or $(\alpha_g-\eps)$-selectable.} under vertex arrivals.
\end{theorem}

\begin{remark}
A classical result of \citet{karp1981maximum} shows that the expected size of the maximum matching of the bipartite Erdős–Rényi random graph $G(n,n,1/n)$ is $\approx0.544 n$ as $n \rightarrow \infty$. This implies that $\approx 0.544$ is an upper bound on the selectability of any RCRS, in fact even any \textit{offline} contention resolution scheme \citep{chekuri2014submodular}, on bipartite graphs.  However, such a result does not contradict our guarantee of $\alpha_\infty\approx0.567$ for bipartite graphs because crucially, we are assuming the vertex-arrival model, in which edges are negatively correlated so that there is at most one active edge incident to each arriving vertex.  By contrast, the upper bound of \citet{karp1981maximum} holds when the edges are independent.

\end{remark}

Our \Cref{thm:recursive_rcrs} improves the state-of-the-art $8/15\approx0.533$-selectable RCRS from \citet{Fu2021} that holds generally under vertex arrivals, except in the case where $G$ has triangles (i.e.~$g=3$).  For this, we design a different RCRS.

\begin{theorem} \label{thm:general_graph_positive}
An RCRS (with polynomial runtime) is $\approx 0.535$-selectable under vertex arrivals.
\end{theorem}
\begin{remark} \label{rem:general_graph_vanishing}
If the graph has \textit{vanishing edges values}, i.e.~$\max_{e \in E}x_e \rightarrow 0$, then the RCRS from \Cref{thm:general_graph_positive} is $\frac{4}{5} - \frac{3 \sqrt{3}}{20} \approx 0.540$-selectable.
\end{remark}

\Cref{thm:general_graph_positive} does not argue by continuous induction and hence does not use our framework, but nonetheless uses the idea of exact selection, as in \citet{Ezra_2020}, and adds it to the algorithm of \citet{Fu2021}.  Our analysis shows that interpolating between the two algorithms performs better than either of them alone, under random-order vertex arrivals.

Finally, we show that an upper bound of $(1+\sfe^{-2})/2$, first proved in \citet{Fu2021} using a vanishing complete graph, still holds when the counterexample is a vanishing complete \textit{bipartite} graph (to our understanding, this requires new technical arguments, as we explain at the end of \Cref{sec:technical_overview}).

\begin{theorem} \label{thm:hardness_vertex_rom}
No RCRS is better than $(1+\sfe^{-2})/2$-selectable for bipartite graphs under vertex arrivals.
\end{theorem}

This shows that our RCRS from \Cref{thm:recursive_rcrs} is \textit{tight}, because $\alpha_{\infty}=(1+\sfe^{-2})/2$.  Tight contention resolution results for matchings in graphs are rare, with some exceptions being \citet{BruggmannZ22} (offline, bipartite, monotone), \citet{Ezra_2020} (adversarial-order, vertex-arrival), and \citet{nuti2023towards} (offline, vanishing edge values).
Complementing \citet{nuti2023towards} who show asymptotic tightness as $\max_{e \in E}x_e\to0$, we show asymptotic tightness in a new and different regime where the odd girth $g\to\infty$, with quantifiably fast convergence in that the error is less than $0.0001$ once $g\ge7$.

\subsection{Techniques for Vertex-batched Arrivals} \label{sec:technical_overview}

Having explained our discretization/sampling framework for continuous induction, we now describe our techniques for establishing the results in \Cref{sec:results}, focusing solely on finding selection functions that satisfy~\eqref{eqn:goal} under the vertex arrival model.

\textbf{Reformulation of arrival process}.  Recall from earlier that in the vertex-batched RCRS problem, every vertex $v$ arrives at a random time, at which point each edge between $v$ and an already-arrived vertex $u$ will be active w.p.~$x_{u,v}$, correlated so that at most one of these edges is active (which is feasible, by~\eqref{eqn:introMatchingPolytope}).
We reinterpret randomness as being realized in the following order instead.
First, every vertex $v$ independently \textit{chooses} each vertex $u$ such that $(u,v)\in E$ w.p.~$x_{u,v}$, choosing at most one vertex whose realization is denoted $F_v$.
Then, every vertex $v$ independently draws an arrival time $Y_v$ uniformly from [0,1], and an edge $(u,v)$ is active if either $F_v=u$ and $Y_u<Y_v$, or $F_u=v$ and $Y_v<Y_u$. In either case, edge $(u,v)$ arrives at time $\max\{Y_u,Y_v\}$.
These two processes for drawing active edges are equivalent because conditional on any realization of arrival times, every arriving vertex $v$ will bring at most one active edge with an already-arrived vertex $u$, drawn with the same probabilities $x_{u,v}$. Reformulating the random process in this way
also allows us to later bound the probability of certain rare events (i.e., event $B_{v \ra u}$ in \Cref{sec:proving_continuous_induction_step}), which is necessary for proving \Cref{thm:recursive_rcrs}.

\textbf{$(1+\sfe^{-2})/2$-selectable RCRS (\Cref{thm:recursive_rcrs} when $g=\infty$)}.
We will show that it is possible to select every edge $(u,v)$ conditional on it arriving at time $y=\max\{Y_u,Y_v\}$ and being active w.p.~$c(y)=(1-\sfe^{-2y})/2y$.
Since $Y_u$ and $Y_v$ are independent random variables drawn uniformly from [0,1], the density of arrival time $\max\{Y_u,Y_v\}$ is $y\mapsto 2y$,  and the guarantee implied by this selection function would be $\int_0^1 c(y) 2y dy=\int_0^1(1-\sfe^{-2y})dy=(1+\sfe^{-2})/2$.
We note that under vertex arrivals, it is conceivable that the selection function $c$ should depend on the values of both arrival times $Y_u$ and $Y_v$, instead of just the value $\max\{Y_u,Y_v\}$.
We show that making $c$ a function of $\max\{Y_u,Y_v\}$ suffices for deriving a tight guarantee on bipartite graphs.

We provide an informal but clean argument that $c(y)=(1-\sfe^{-2y})/2y$ satisfies~\eqref{eqn:goal} when $y=\max\{Y_u,Y_v\}$.  Suppose an edge $e=(u,v)$ arrives at time $t\in[0,1]$, where we assume without loss that $Y_u<Y_v=t$.  Under vertex arrivals, edge $e$ is feasible to accept if and only if vertex $u$ was not matched before time $t$.  Thus, we must show that
\begin{align} \label{eqn:goalVertexArrival}
\bP[u\text{ not matched before }t \mid Y_u<Y_v=t]\ge c(t),
\end{align}
assuming the induction hypothesis that active edges arriving at times $y$ before $t$ were accepted with probability $c(y)$.
Like in \Cref{sec:introContInd}, we need to drop the conditioning on the left-hand-side of~\eqref{eqn:goalVertexArrival} in order to apply the induction hypothesis.  However, under vertex arrivals, this no longer follows from a simple argument, and in fact dropping the conditioning could, unfortunately, increase the left-hand-side of~\eqref{eqn:goalVertexArrival}.

Our solution is to drop the conditioning on $Y_v=t$, but not on $Y_u<t$, which is sufficient for applying the induction hypothesis.  On bipartite graphs, we prove (for the complement events) that
\begin{align} \label{eqn:dropConditioning_vertex}
\bP[u\text{ is matched before }t|Y_u<Y_v=t]\le\bP[u\text{ is matched before }t|Y_u<t],
\end{align}
using a refined\footnote{In our formal proofs, the induction hypothesis is that an active edge $e=(u,v)$ arriving at time $y$ is accepted w.p.~$c(y)$ \textit{further conditioned} on either $Y_u<Y_v=y$ or $Y_v<Y_u=y$.  Our informal proof sketches omit this detail.} induction hypothesis.
We can then analyze the right-hand-side of~\eqref{eqn:dropConditioning_vertex}:
\begin{align*}
\bP[u\text{ is matched before }t \mid Y_u<t]
&=\sum_w t \cdot x_{u,w}\bP[(u,w)\text{ is accepted before }t \mid Y_u<t,Y_w<t,(u,w)\text{ is active}]
\\ &=\sum_w t\cdot x_{u,w}\int_0^t c(y) \frac{2y}{t^2} dy
\\ &\le \frac1t \int_0^t \frac{1-\sfe^{-2y}}{2y} 2y dy=\frac1t \left(t-\frac{1-\sfe^{-2t}}{2} \right)=1-c(t)
\end{align*}
where the first equality holds because the events of $(u,w)$ being accepted for different vertices $w$ are mutually exclusive,
the second equality applies the induction hypothesis (where the conditional arrival time of an edge $(u,w)$ has probability density $y\mapsto 2y/t^2$),
and the inequality applies~\eqref{eqn:introMatchingPolytope}.
This in conjunction with~\eqref{eqn:dropConditioning_vertex} establishes the desired~\eqref{eqn:goalVertexArrival} for bipartite graphs.

\textbf{Generalizing to odd-girth $g \ge 3$ via correlation decay (\Cref{thm:recursive_rcrs}).}
When $G$ is not bipartite,~\eqref{eqn:dropConditioning_vertex} no longer holds.  For example, if vertices $u,v,w$ form a 3-cycle, then allowing for the case where $Y_v<t$ (equivalently, dropping the conditioning on $Y_v=t$) could match $v$ with $w$ before time $t$, decreasing the probability that $u$ is matched (to $w$) before $t$.  For general graphs, we prove a weaker version where
\begin{equation} \label{eqn:simple_correlation_decay}
 \mb{P}[\text{$u$ is matched before $t$}| Y_u < Y_v= t] - \mb{P}[\text{$u$ is matched before $t$}| Y_u < t] \le \frac{1}{t} \int_{0}^{t} \frac{2 y^{g-1}}{(g-1)!} dy.
\end{equation}
Recall that $(u,v)\in E$, so a path of even length between $u$ and $v$ must have length at least $g-1$.
\Cref{eqn:simple_correlation_decay} can be viewed as a \textit{correlation decay} argument, where the impact of arrival time $Y_v$ on $u$ being matched, upper-bounded by term $\frac{1}{t} \int_{0}^{t} \frac{2 y^{g-1}}{(g-1)!} dy$, vanishes to 0 as the distance between $v$ and $u$ increases when $g\to\infty$.

We only worry about even paths (and hence odd-girth) because we show that 
in order for allowing $Y_v < t$ to decrease
the probability that $u$ is matched before time $t$, roughly speaking, both of the following must occur:
\begin{itemize}
    \item An even path $P$ of \textit{active} edges between $u$ and $v$ exists (in the preceding example, the even path has length 2, consisting of edges $(u,w)$ and $(w,v)$);
    \item The edges of $P$ arrive in a specific order (where allowing $Y_v<t$ would cause $(w,v)$ to be matched instead of $(u,w)$).
\end{itemize}
Since the vertices arrive in random order, and there are at most $|V|$ active edges in $G$, it is unlikely that any such path $P$ exists, leading to \eqref{eqn:simple_correlation_decay}.

Interestingly, our argument here
allows an adversary to set the "choices" $(F_v)_{v \in V}$ of every vertex $v$ under our reformulated arrival process.
Indeed, our correlation decay argument does not take an expectation over the realizations of $F_v$ from the probabilities $x_e$ (as done in other correlation decay arguments, e.g.~\citet{kulkarni2022}), and already leads to very fast convergence (in $g$) to the tight guarantee $(1+\sfe^{-2})/2$.
Our analysis is also tight, once we concede to an adversarial choice of $(F_v)_{v\in V}$.

\textbf{RCRS not dependent on odd-girth (\Cref{thm:general_graph_positive}).}
Our RCRS from \Cref{thm:recursive_rcrs} is not better than $1/2$-selectable on general graphs, due to the rate at which it selects
arriving edges $e$ with $Y_e$ close to $1$. In particular, \eqref{eqn:simple_correlation_decay} provides
a poor upper bound on the difference between conditioning vs.\ not conditioning on $Y_v=t$, as $t$ gets closer $1$. In contrast, when $t$ is close to $0$, \eqref{eqn:simple_correlation_decay} provides a good upper bound on the difference no matter the value of $g$. As a result, our $g=3$ RCRS selects edges $e$ with $Y_e$ close to $0$ at roughly the
same rate as our bipartite RCRS. Since our bipartite RCRS is tight, this suggests that in order to improve on the general graph RCRS of \citet{Fu2021},
one should define a \textit{hybrid} or \textit{two-phased} RCRS which executes our $g=3$ RCRS for vertices which arrival before some time $t\in[0,1]$, followed
by the RCRS of \citet{Fu2021}. Unfortunately, analyzing such an RCRS does not seem possible in the second phase.
This is because the analysis of the RCRS of \citet{Fu2021} requires conditioning on multiple vertices' arrival times, which then destroys
the (unconditional) exact selection guarantee of our $g=3$ RCRS.  

As a solution, we again use the idea of exact selection, however this time by explicitly using the vertex-arrival OCRS from \citet{Ezra_2020} in the first phase. Crucially, this OCRS accepts every active edge with probability \textit{exactly} 1/2, conditional on \textit{any} realization of arrival times.  This is a very strong guarantee which allows us to successfully analyze the RCRS of \citet{Fu2021} in the second phase. Our RCRS interpolates between the RCRS of \cite{Fu2021} and the OCRS of \cite{Ezra_2020} by setting $t=0$ and $t=1$, respectively. By taking $0 < t < 1$, we get a \textit{hybrid} RCRS which outperforms both schemes. Intuitively, this is because the \textit{non-greedy} decisions made in phase $1$ allow us to more equitably select edges in the second phase. By optimizing over $t$, we attain a selectability of $\approx 0.535$.

\textbf{Upper bound for bipartite graphs (\Cref{thm:hardness_vertex_rom}).} 
In order to prove the $(1+\sfe^{-2})/2$ negative result of \Cref{thm:hardness_vertex_rom},
we consider the complete bipartite graph on $2n$ vertices with uniform edge values
of $1/n$. We argue that any RCRS can match at most $(1+ o(1))(1+\sfe^{-2}) n/2$ edges in expectation as $n \rightarrow \infty$.
Our approach is similar to \citet{Fu2021} in that we bound the one-step changes of the RCRS and then apply
the differential equation method \cite{de, warnke2019wormalds, bennett2023}.
However, since \citet{Fu2021} analyze a complete graph, all arrival orders are equivalent, by symmetry.
In contrast, in the bipartite graph we analyze, we must argue that in a \textit{typical} random arrival order, the RCRS can only gain $o(n)$ from leveraging the asymmetry in the number of vertices that have arrived on each side of the graph.
Put another way, the upper bound of \citet{Fu2021} holds even if the algorithm could choose the order of arrival, but our upper bound for bipartite graphs must\footnote{On bipartite graphs, if the algorithm could choose the vertex arrival order, then in fact a guarantee of $1-1/\sfe$ is possible---just choose one side of the graph (the "offline" side) to arrive entirely before the other, and then run a $(1-1/\sfe)$-selectable rank-1 RCRS for each vertex on the offline side.} exploit the fact that a uniform random arrival order does not help the algorithm too much.
After making this argument, we deduce that the upper bound of $(1+\sfe^{-2})/2$ holds even if we restrict the graph to be bipartite.

\subsection{Further Discussion of Related Work}

\textbf{RCRS for rank-1 matroids and edge arrivals}.
\citet{Lee2018} first established a $(1-1/\sfe)$-selectable RCRS for rank-1 matroids, which is \textit{tight} in that a guarantee $\alpha>1-1/\sfe$ is not possible for arbitrary rank-1 matroids and activeness probabilities.
\citet{brubach2021offline} establish alternate $(1-1/\sfe)$-selectable RCRS's for rank-1 matroids, in which the attenuation function is agnostic to the arrival time.
They also first established a $(1-\sfe^{-2})/2$-selectable RCRS for random-order edge arrivals in general graphs, using the same time-agnostic attenuation functions.
\citet{pollner2022improved} establish an alternate $(1-\sfe^{-2})/2$-selectable RCRS using the time-aware attenuation function of \citet{Lee2018}, and also break the barrier of $(1-\sfe^{-2})/2$ for this problem.  The state-of-the-art for this problem of random-order edge arrivals in general graphs (again using time-agnostic attenuation functions) is currently given by \citet{macrury2023random}.

\textbf{Vertex-batched arrivals}.
Motivated by prophet inequality problems on graphs, contention resolution schemes under vertex-batched arrivals were introduced in \citet{Ezra_2020}.  They derive a tight $1/2$-selectable OCRS, which assumes adversarial vertex arrival order. For the query-commit matching \citep{Chen} and price of information \citep{singla2018price} problems, 
\citet{Fu2021} observed that attaining an $\alpha$-approximation guarantee for either problem reduces to designing
an $\alpha$-selectable RCRS for vertex arrivals on
a general graph, so their 8/15-selectable RCRS implies an $8/15$-approximation for each problem. A similar reduction was observed in \citet{borodin2023} for the query-commit matching problem on bipartite graphs with patience constraints on one side of the bipartition: each vertex \textit{without} patience constraints executes its own rank-$1$ $(1-1/\sfe)$-selectable RCRS,
leading to a $(1-1/\sfe)$-approximation. We improve the results of \cite{Fu2021}, and hence also the best-known approximation guarantees for the query-commit matching \citep{Chen} and price of information \citep{singla2018price} problems on general graphs. Recently, RCRS for vertex arrivals have also been
used in the stochastic weighted matching problem when $O(1/p)$ edges per vertex can be queried,
where $p$ is the minimum edge probability of the stochastic graph. \citet{derakhshan2023query} prove an interesting connection that reduces designing a $1/(2-\alpha)$-approximation algorithm for the stochastic matching problem to designing an $\alpha$-selectable RCRS for vertex arrivals. Since we improve on \cite{Fu2021}, our new selectability bounds are also relevant here.

\textbf{Sampling to achieve exact selection}.
Our algorithm requires sampling of its past executions in order to approximate its stochastic state at any point in time.  This idea has been previously used for other problems such as stochastic knapsack \citep{ma2014improvements} and query-commit matching \citep{Adamczyk15}, and is sometimes summarized as "simulation-based attenuation" in the literature \citep{BrubachSSX20}.  Importantly, the algorithm should not \textit{re-sample} previously-recorded state approximations, in order to avoid error propagation \citep{ma2014improvements}.  We use this idea in our framework, which has to be carefully combined with the discretization of a state space that evolves in continuous time.  To our knowledge, previous papers have always had discrete time steps, and hence have not encountered this challenge of "exact selection in continuous time".

\section{Notation and Preliminaries}

We formalize notation for our problem, recapping some concepts from the introduction.

Let $G=(V,E)$ be a graph.  An edge $e=(u,v)$ is said to be \textit{incident} to vertices $u$ and $v$, and $v$ is said to be a \textit{neighbor} of $u$ (and vice versa).
For any vertex $v\in V$, let $N(v)\subseteq V$ denote its neighbors, and let $\partial(v)\subseteq E$ denote the set of edges incident to $v$.
For any $e =(u,v)\in E$, let $\partial(e) := \partial(u) \cup \partial(v) \setminus \{e\}$ denote the set of edges "incident" to $e$.
A \textit{matching} $M$ is a subset of edges no two of which are incident to the same vertex, i.e.\ satisfying $|M\cap\partial(v)|\le 1$ for all $v\in V$.  A vector $\bm{x} = (x_e)_{e \in E} \in[0,1]^E$ lies in the \textit{matching polytope} of $G$ if $\sum_{e\in\partial(v)}x_e\le 1$ for all $v\in V$. In this case,
we refer to $\bm{x}$ as a \textit{fractional matching} for $G$, and $\bm{x}$ is said to be \textit{1-regular} if $\sum_{e\in\partial(v)}x_e=1$ for all $v\in\ V$.  We emphasize that generally our analysis does not require $\bm{x}$ to be 1-regular, except \Cref{sec:general_graph_positive}, in which assuming $1$-regularity is without loss due to the reduction of \cite{Fu2021}.

Fix a fractional matching $\bm{x}$. For each $v \in V$, draw $F_v \in N(v) \cup\{\perp\}$ independently, where $\mb{P}[ F_v = u] = x_{u,v}$ for each $u \in N(v)$, and $\mb{P}[F_v = \perp] = 1- \sum_{u \in N(u)} x_{u,v}$
We say that $v$ \textit{chooses} $u$ if $F_v =u$. Here $\perp$ is a \textit{null} element
indicating that $v$ does not choose any vertex. In the \textit{random-order vertex arrival model}, we assume that each $v$ has an \textit{arrival time} $Y_v$ drawn independently and uniformly from $[0,1]$. 
We assume that all arrival times are distinct, which occurs w.p.~1.
We say that $e =(u,v) \in E$ is \textit{active}, provided $F_v =u$ and
$Y_u < Y_v$, or $F_u =v$ and $Y_v < Y_u$.  Note that $\mb{P}[\text{$e$ is active}] = x_e$ for each $e \in E$. 
It will be convenient to define $Y_e = \max\{Y_u, Y_v\}$ to be the \textit{arrival time} of the edge $e$.

Fix a graph $G$ with a fractional matching $\bm{x} = (x_e)_{e \in E}$.
In order to define our algorithms, we first introduce the some additional generic terminology for an RCRS. Given $e=(u,v) \in E$, we say that $v$ \textit{selects} $u$ (respectively, $u$ selects $v$), provided $Y_u < Y_v$ (respectively, $Y_v < Y_u$) and $e$ is matched by the RCRS. Moreover, given $y \in [0,1]$,
if $v$ selects $u$ and $Y_u < Y_v \le y$, then we say that $v$ selects $u$ \textit{by time} $y$. 
We say that $u$ is \textit{matched by time} $y$, if there exists some $v \in N(u)$ such that $u$ selects $v$ or $v$ selects $u$ (with either event occurring by time $y$). Finally, if a vertex $u$ is \textit{not} matched by time $y$,
then we say that $u$ is \textit{safe at time} $y$. Using this terminology, we introduce the following indicator random variables:
\begin{itemize}
\item $M_{v \rightarrow u}(y) =1$ if and only if $v$ selects $u$ by time $y$;
\item $M_{u}(y)=1$ if and only if $u$ is matched by time $y$;
\item $S_u(y) =1$ if and only if $u$ is safe at time $y$.
\end{itemize}
Observe that $M_{u}(y) = \sum_{v \in N(u)} ( M_{v \ra u}(y) + M_{u \ra v}(y))$ and $S_{u}(y) = 1 - M_{u}(y)$.

\subsection{Selection Function for General Odd-girth}

Suppose that $G$ has odd-girth $g \ge 3$, where $g := \infty$ if $G$ has no odd-cycles (and thus is
bipartite). We refer to $c(y): [0,1] \rightarrow [0,1]$
as a \textit{selection function}, provided the following conditions hold:
\begin{enumerate}
\item $c$ is decreasing on $[0,1]$.
\item There exists a constant $C > 0$,
such that $c(y) \ge C$ for all $y\in[0,1]$.
\item For any $t \in[0,1]$, we have that
\begin{align} \label{eqn:designOfC}
c(t) \le 1-\frac1t\int_0^t 2 (c(y) y  +\phi_g(y)) dy,
\end{align}
where $\phi_g(y):= \frac{y^{g-1}}{(g-1)!}$, and $\phi_g(y):=0$ if
$g = \infty$. 
\end{enumerate}
The key condition is~\eqref{eqn:designOfC}, which was the goal of our framework as described in the introduction in~\eqref{eqn:goal}.  The first two conditions are minor technicalities useful for discretization and sampling respectively, which are naturally satisfied when we solve the differential equation implied by~\eqref{eqn:designOfC}.

We now derive in closed form the selection function to be used for a graph with given odd-girth $g$.
Let $c(t)$ be the general solution to the differential equation
\begin{equation} \label{eqn:differential_selection}
    c'(t) = \frac{(1 - 2 (t c(t) + \phi_g(t)) - c(t))}{t}
\end{equation}
After solving for \eqref{eqn:differential_selection}, we get a closed form solution in terms of the \textit{upper incomplete gamma function} $\Gamma(s,z):= \int_{z}^{\infty} \zeta^{s-1} e^{-\zeta} d\zeta$ for $s > 0$ and $z \in \mb{R}$. (Here $\Gamma(0,z)$ is the usual \textit{gamma function}). By choosing the constant of integration appropriately, we get the following
function
\begin{equation} \label{eqn:optimal_selection_function}
c(y) = \frac{1}{2y} - \frac{\sfe^{-2 y}}{2y} - \frac{\sfe^{-2 y} ( \Gamma(g, -2y) - \Gamma(g, 0))}{2^{g-1} y(g -1)!},
\end{equation}
where $c(0):=1$, and $\lim_{y \rightarrow 0^+} c(y) =1$. Clearly, $c$ is a selection function, as $c(0)=1$, $c(1) > 0$, and $c$ is decreasing since
the right-hand side of \eqref{eqn:differential_selection} is negative for all $t \in (0,1]$. The integral equation \eqref{eqn:designOfC} can be seen to hold
after rearranging \eqref{eqn:differential_selection} and integrating.

For each odd $g \ge 3$, define
\begin{equation} \label{eqn:alpha_definition}
    \alpha_g := 2 \int_{0}^{1} c(y) y dy,
\end{equation}
where $\alpha_{\infty}:= (1+\sfe^{-2})/2$.
There exists a closed form solution to \eqref{eqn:alpha_definition}
in terms of the incomplete gamma function. We state this in \Cref{appendix} in expression \eqref{eqn:gamma_explicit}, and
also prove the following:
\begin{proposition} \label{prop:increasing_alpha}
    $(\alpha_{2k+1})_{k=1}^{\infty}$ is an increasing sequence which converges to $\alpha_{\infty}$.
\end{proposition}

\section{General Framework for Proving \Cref{thm:recursive_rcrs}} \label{sec:recursive_rcrs}

Let us suppose that $G$ has odd-girth $g \ge 3$, and $c$ is an arbitrary selection function.
Fix $T \in \mb{N}$ to be a \textit{discretization constant}, and $\delta \in [0,1)$ to be a \textit{sampling parameter}. Note that we can take $\delta = 0$ when proving the \textit{existence} of an RCRS with the selectability claimed in \Cref{thm:recursive_rcrs} (and ignore issues of computational efficiency). We define an RCRS for each $T$ and  $\delta$, and show
that its selectability on $G$ is $(1 - o_{T,\delta}(1))\int_{0}^{1} 2 c(y) y dy$
(here $o_{T,\delta}(1)$ is a function which tends to $0$ as $T \rightarrow \infty$ and $\delta \rightarrow
 0$).  \Cref{thm:recursive_rcrs} then follows by taking $c$ as in \eqref{eqn:optimal_selection_function}.

For each $j=0, \ldots ,T$, define $t_j=j/T$.
We divide $(0,1]$ into $T$ intervals, $I_j:=(t_{j-1},t_j]$ for each $1 \le j \le T$. 
We define our RCRS in $T$ \textit{phases}, where in phase $j$ the RCRS
processes vertices $v$ which arrive in the interval $I_j$ (i.e., $Y_v \in I_j$).  We ignore the measure-zero event of an arrival time being $0$, so that all vertex arrivals occur during a phase. Our RCRS is defined
recursively over indices $j=0,1,\ldots,T-1$.  To elaborate, assuming we've defined an RCRS for vertices which arrive up until
the end of phase $I_j$ (time $t_{j}$), we extend the RCRS to vertices which arrive in the interval $I_{j+1}$, thus defining an RCRS for
vertices which arrive up until time $t_{j+1}$. 

Suppose that $v$ arrives at time $Y_v = y$ where $y \in I_{j+1}$.
Assume that $F_v = u$ and $Y_u < y$ so that $(u,v)$ is active. At this point, we would like to compute $\bE[S_u(t_{j}) \mid Y_u<t_{j}<Y_v]$, which is the probability that $u$ is safe at time $t_{j}$, given $Y_u<t_{j}<Y_v$ (here $\bE[S_u(t_{0}) \mid Y_u<t_{0}<Y_v] :=1$). Observe that this probability is well-defined, as it is determined by the decisions of the RCRS up until time $t_j$. However, since we want our RCRS to be polynomial-time, exact computation is not possible, and so we must we instead estimate these values: suppose that the \textit{estimates} $\hS(i):= (\hS_{v \ra u }(i))_{u,v \in V}$ were computed for $i=0, \ldots ,j-1$ in the previous phases,
where $\hS_{v \ra u }(0):=1$ for all $u,v \in V$. Then, we compute an estimate $\hS_{v \ra u }(j)$
of $\bE[S_u(t_{j}) \mid Y_u<t_{j}<Y_v, (\hS(i))_{i=0}^{j-1}]$
via Monte-Carlo sampling by simulating the execution of
the RCRS up until time $t_j$. Specifically, we take $Q \in \mb{N}$ independently drawn executions of the RCRS up until time $t_j$ and define $\hS_{v \ra u }(j)$ to
be the fraction of executions where $u$ was safe, i.e.\ not matched. (Here the specific value of $Q$ is defined later in \Cref{lem:sampling_step}.) Crucially, in these simulations, the previously-computed estimates $(\hS(i))_{i=0}^{j-1}$ are used, and the arrival time of $u$ (respectively $v$) is drawn uniformly from $[0,t_j)$ (respectively, $(t_j,1]$). We emphasize that we compute these directed estimates $\hS_{v \ra u }(j)$ and $\hS_{u \ra v}(j)$ for all edges $(u,v)\in E$, irrespective of whether $v$ actually arrived during $I_{j+1}$ with $F_v=u$, because these estimates are needed for Monte-Carlo sampling in future phases.
We say that $\hS_{v \ra u }(j)$ is a $\delta$-\textit{approximation} of $\bE[S_u(t_{j}) \mid Y_u<t_{j}<Y_v, (\hS(i))_{i=0}^{j-1}]$ if
\begin{equation}\label{eqn:IHsamplingError}
1-\delta \le \frac{\hS_{v \ra u }(j)}{\bE[S_u(t_{j}) \mid Y_u<t_{j}<Y_v, (\hS(i))_{i=0}^{j-1}]}\le 1+\delta. 
\end{equation}
If we do not care about computational efficiency, then we can
take $\delta =0$ and \textit{define} $\hS_{v \ra u }(j)$ to be
$\bE[S_u(t_{j}) \mid Y_u<t_{j}<Y_v]$ for each $j=0, \ldots , T-1$
and $u,v \in V$.

Having computed the estimate $\hS_{v \ra u}(j)$, the RCRS draws $A_{v \ra u}$ from $\Ber \left(\min\left( \frac{c(y)}{\hS_{v \ra u }(j)}\frac{1-\delta}{1+\frac1{CTy}}, 1\right) \right)$ independently (technically, since $\hS_{v \ra u}(j)$ is sampled, the Bernoulli parameter here is a random variable, however the reader should think 
of $\frac{c(y)}{\hS_{v \ra u }(j)}\frac{1-\delta}{1+\frac1{CTy}} \approx \frac{c(y)}{\bE[S_u(t_{j}) \mid Y_u<t_{j}<Y_v]}$). If $A_{v \ra u}=1$ and $u$ is safe at time $y$,
then $e$ is added to the matching. The RCRS then processes the remaining vertices in phase
$j+1$ before continuing to the next phase. Below is a formal description of the $T$ phases
of the RCRS.

\begin{algorithm}[H]
\caption{Recursive RCRS}
\label{alg:recursive_rcrs}
\begin{algorithmic}[1] 
\Require $G=(V,E)$, $\bm{x}=(x_e)_{e \in E}$, $T,Q \in \mb{N}$, and $\delta \in (0,1)$.
\Ensure subset of active edges forming a matching $\scr{M}$
\State $\scr{M} \leftarrow \emptyset$

\For{$j=0, \ldots ,T-1$}

\If{$j=0$}
\State Set $\hS_{v \ra u }(0) =1$ for each $u, v \in V$.

\ElsIf{$j \ge 1$}
\State For each $u, v \in V$, using $Q$ independently drawn executions of the RCRS up until time $t_j$, compute an estimate $\hS_{v \ra u }(j)$ of $\bE[S_u(t_{j}) \mid Y_u<t_{j}<Y_v, (\hS(i))_{i=0}^{j-1}]$.

\EndIf

\For{arriving vertices $v$ with $Y_v \in I_{j+1}$, and $F_v \neq \perp$}
\State Set $F_v = (u,v)$ and $y = Y_v$.

\State Draw $A_{v \ra u}$ from $\Ber \left(\min \left( \frac{c(y)}{\hS_{v \ra u }(j)}\frac{1-\delta}{1+\frac1{CTy}}, 1\right) \right)$ independently.
\If{$u$ is unmatched, $A_{v \ra u}=1$, and $Y_{u} < y$}
\State $\scr{M}\leftarrow \scr{M} \cup \{(u,v)\}$
\EndIf

\EndFor
\EndFor

\State \Return $\scr{M}$.

\end{algorithmic}
\end{algorithm}
\begin{remark}
The estimates $\hS(j) := (\hS_{v \ra u }(j))_{u, v \in V}$ are computed once at the start of the phase, and then never updated again. Moreover, since the estimates $(\hS(j))_{=0}^{T-1}$ are computed using independent executions of \Cref{alg:recursive_rcrs} on copies of $G$, they are independent from the random variables $(F_v)_{v \in V}$ and $(Y_v)_{v \in V}$ of $G$.
\end{remark}

For each $0 \le j \le T-1$, observe that \Cref{alg:recursive_rcrs} uses the randomly drawn estimates $\hS(j) = (\hS_{v \ra u }(j))_{u, v \in V}$ in phase $j+1$, even when the estimates happen to be \textit{bad} (i.e., they aren't $\delta$-approximations). Moreover, if the estimates are bad, then we can't easily control the $A_{v \ra u}$ random variables,
which makes analyzing the RCRS challenging. Instead, we only analyze the algorithm when the estimates are good,
and then argue that the estimates are good with sufficiently high probability. Our goal for now is to prove that $\bE[M_{u \ra v}(y) \mid Y_u<Y_v=y,F_v=u] = (1 - o_{T,\delta}(1)) c(y)$ for all edges $(u,v)$ and all arrival times $y\in(0,1]$, assuming the sampled estimates are always good.

More formally, let $B_j$ be the (bad) event in which $\hS_{v \ra u }(j)$
is \textit{not} a $\delta$-approximation for some $u,v \in V$. Observe that
$B_j$ is determined by $(\hS(i))_{i=0}^j$ (i.e., $B_j$ 
is $(\hS(i))_{i=1}^j$-measurable), and so given $(\hS(i))_{i=0}^j$, we can always check whether an \textit{instantiation} $\hs(j)$ of $\hS(j)$
satisfies $B_j$ or not.
We first analyze \Cref{alg:recursive_rcrs} in phase $j +1$ when $(\hS(i))_{i=0}^j$ satisfies $\cap_{i=0}^{j} \neg B_i$.
Afterwards, we show that this is enough to prove \Cref{thm:recursive_rcrs},
as we argue that $\mb{P}[ \cup_{i=0}^{T-1} B_i] \le \delta$ assuming $Q$ is sufficiently large (see \Cref{lem:sampling_step}). Let us now consider the following hypothesis, dependent on $j=0, \ldots , T-1$, which we shall afterwards prove inductively: 

For each instantiation $(\hs(i))_{i=0}^j$ of $(\hS(i))_{i=0}^j$ which satisfies $\cap_{i=0}^{j} \neg B_i$, and each $u,v \in V$, $y \in I_{j+1}$,
\begin{align} \label{eqn:IHmatchingProb}
\frac{1-\delta}{1+\delta}\left( 1-\frac{4}{CTy} \right) c(y) \le \bE[M_{v \ra u}(y) \mid Y_u<Y_v, Y_v=y,F_v=u,
(\hS(i))_{i=0}^{j} = (\hs(i))_{i=0}^j
]\le c(y).
\end{align}

\begin{lemma}[Base Case] \label{lem:base_case}
For $j=0$, \eqref{eqn:IHmatchingProb} holds.
\end{lemma}
\begin{proof}
First observe that when $j =0$, each $\hS_{v \ra u }(0)$ is a $\delta$-approximation by definition, 
and so we can ignore
the conditioning on $(\hS(i))_{i=0}^j = (\hs(i))_{i=0}^{j}$ in \eqref{eqn:IHmatchingProb}.
Now, observe that if $y\in I_1$, then
\begin{align*}
\bE[M_{v\to u}(y) \mid Y_u<Y_v=y,F_v=u]
&=\bE[S_u(y) \mid Y_u<y<Y_v]c(y)\frac{1-\delta}{1+\frac1{CTy}}.
\end{align*}
As such, the upper bound of \eqref{eqn:IHmatchingProb} holds.  We claim that $\bE[S_u(y) \mid Y_u<y<Y_v]\ge1-2/T$ holds for the lower bound.
To see this, let us first condition on $Y_u <y < Y_v$. Observe
then that in order for $S_u(y)=0$, we would need another vertex $w \in N(u) \setminus \{v\}$
to arrive before time $y$ and match to $u$. This occurs with probability
at most $2 y$, as either $w$ selects $u$, or $u$ selects $w$, and there is at most one such $w$ in expectation (recall that $\sum_w x_{u,w}\le1$). Since $y \in I_1$,
this probability is at most $2/T$, which justifies the above equation.
This completes the base case because
$$
\frac{1-2/T}{1+1/(CTy)}\ge \frac{1-3/(CTy)}{1+1/(CTy)}\ge1-\frac4{CTy}.
$$
\end{proof}

For a time step $1 \le k \le T-1$, we first apply strong induction assuming~\eqref{eqn:IHmatchingProb} holds for all $0\le j\le k-1$, to establish the following lower bound on vertices being safe at time $t_k$.

\begin{lemma} \label{lem:continuous_induction_step}
Fix $1 \le k \le T-1$. Suppose that \eqref{eqn:IHmatchingProb} holds
for all $0 \le j \le k-1$.
Then for each instantiation $(\hs(i))_{i=0}^{k-1}$ of $(\hS(i))_{i=0}^{k-1}$ which satisfies $\cap_{i=0}^{k-1} \neg B_i$, we have that
$\bE[S_u(t_k) \mid Y_u<t_k<Y_v, (\hS(i))_{i=0}^{k-1} = (\hs(i))_{i=0}^{k-1}] \ge c(t_k)$
for each $u,v \in V$. 
\end{lemma}

\Cref{lem:continuous_induction_step} is stated in such a way that we can ignore the technicalities caused by the discretization and sampling errors: it boils down to proving $\bE[S_u(t_k) \mid Y_u<t_k<Y_v]\ge c(t_k)$ for all $(u,v)$ (i.e.\ any vertex $u$ is "safe" for any other vertex $v$ to select at time $t_k$ w.p.~at least $c(t_k)$), assuming $\bE[M_{v\to u}(y)\mid Y_u<Y_v=y,F_v=u]\le c(y)$ for all $(u,v)$ and $y\le t_k$ (i.e.\ vertices did not get matched with too high of probabilities before time $t_k$, given by the upper bound of~\eqref{eqn:IHmatchingProb}).
This forms the core of our analysis, as outlined in the ``continuous induction'' argument in \Cref{sec:technical_overview}. We defer the proof of \Cref{lem:continuous_induction_step} to \Cref{sec:proving_continuous_induction_step}.

We now complete the induction assuming \Cref{lem:continuous_induction_step} is true.
This induction was unnecessary under the informal "continuous induction" argument from \Cref{sec:technical_overview}, as there was no error in the bounds
of \eqref{eqn:IHmatchingProb} (i.e., $\bE[M_{u \ra v}(y) \mid Y_u<Y_v=y,F_v=u] = c(y)$)
because it took $T=\infty$ and $\delta=0$.

\begin{lemma}[Inductive Step] \label{lem:discretization_step}
Fix $1 \le k \le T-1$. Suppose that \eqref{eqn:IHmatchingProb} holds
for all $0 \le j \le k-1$.  Then~\eqref{eqn:IHmatchingProb} holds for $j=k$.
\end{lemma}

\begin{proof}[Proof of \Cref{lem:discretization_step}]
Fix $1 \le k \le T-1$, $y\in I_{k+1}$, $u, v \in V$, and an instantiation $(\hs(i))_{i=0}^k$ of $(\hS(i))_{i=0}^k$ which satisfies $\cap_{i=0}^k \neg B_i$. Now, since $\hs(k)$ satisfies $\neg B_k$,
we know that $\hs_{v \ra u }(k)$ is a $\delta$-approximation. That is,
\begin{equation} \label{eqn:delta_approx}
1-\delta \le \frac{\hs_{v \ra u }(k)}{\bE[S_u(t_{k}) \mid Y_u<t_{k}<Y_v, (\hS(i))_{i=0}^k = (\hs(i))_{i=0}^{k-1}]}\le 1+\delta.
\end{equation}
Applying \Cref{lem:continuous_induction_step} we have  $\bE[S_u(t_k) \mid Y_u<t_k<Y_v, (\hS(i))_{i=0}^{k-1} = (\hs(i))_{i=0}^{k-1}] \ge c(t_k)$, and since $c(y) \le c(t_k)$
as $c$ is decreasing, we know that
$
\frac{c(y)}{\hs_{v \ra u }(k)}\frac{1-\delta}{1+\frac1{CTy}} \le 1.
$
Thus, $A_{v \ra u} \sim \Ber \left(\frac{c(y)}{\hs_{v \ra u }(k)}\frac{1-\delta}{1+\frac1{CTy}} \right)$
(conditional on $(\hS(i))_{i=0}^{k-1} = (\hs(i))_{i=0}^{k-1}$).

For the remainder of the proof, we implicitly condition on
$(\hS(i))_{i=0}^{k-1} = (\hs(i))_{i=0}^{k-1}$ to save space.
Now, by the definition of \Cref{alg:recursive_rcrs},
\begin{align*}
\bE[M_{v\to u}(y) \mid Y_u<Y_v, Y_v=y,F_v=u]
&=\bE[S_u(y) \mid Y_u<y<Y_v] \cdot \bE[A_{v \ra u}] \\
&= \bE[S_u(y) \mid Y_u<y<Y_v]\frac{c(y)}{\hs_{v \ra u }(k)}\frac{1-\delta}{1+\frac1{CTy}}.
\end{align*}
Moreover, by applying each bound of \eqref{eqn:delta_approx}, 
$$
\bE[M_{v\to u}(y) \mid Y_u<Y_v=y,F_v=u]
\le\frac{\bE[S_u(y) \mid Y_u<y<Y_v]}{\bE[S_u(t_k) \mid Y_u<t_k<Y_v]}\frac{c(y)}{1+\frac1{CTy}},
$$
and
$$
\bE[M_{v\to u}(y)\mid Y_u<Y_v=y,F_v=u] \ge \frac{\bE[S_u(y)\mid Y_u<y<Y_v]}{\bE[S_u(t_k)\mid Y_u<t_k<Y_v]}\frac{1-\delta}{1+\delta}\frac{c(y)}{1+\frac1{CTy}}.
$$
Thus, to complete the proof, we must upper and lower bound the ratio $\frac{\bE[S_u(y)\mid Y_u<y<Y_v]}{\bE[S_u(t_k)\mid Y_u<t_k<Y_v]}$ by
$1+\frac{1}{CTy}$ and $1-\frac{3}{CT}$, respectively.
For the upper bound, since $S_u(y)\le S_u(t_k)$ for $y > t_k$, and $c$ is a selection function, $\bE[S_u(t_k)\mid Y_u<t_k<Y_v]\ge c(t_k)\ge C$.
\begin{align*}
\frac{\bE[S_u(y)\mid Y_u<y<Y_v]}{\bE[S_u(t_k)\mid Y_u<t_k<Y_v]}
&=\frac{t_k}{y}\frac{\bE[S_u(y)\mid Y_u<t_k<Y_v]}{\bE[S_u(t_k)\mid Y_u<t_k<Y_v]}+(1-\frac{t_k}y)\frac{\bE[S_u(y)\mid t_k\le Y_u<y<Y_v]}{\bE[S_u(t_k)\mid Y_u<t_k<Y_v]}
\\ &\le\frac{t_k}{y}+(1-\frac{t_k}y)\frac1C \le1+\frac{1}{CTy}
\end{align*}
where the final inequality uses that $t_k\ge y-1/T$.
The lower bound is more involved. We state the sequence of inequalities,
and provide an explanation directly following:
\begin{align*}
\frac{\bE[S_u(y)\mid Y_u<y<Y_v]}{\bE[S_u(t_k)\mid Y_u<t_k<Y_v]}
&=\frac{t_k}{y}\frac{\bE[S_u(y)\mid Y_u<t_k<Y_v]}{\bE[S_u(t_k)\mid Y_u<t_k<Y_v]}+(1-\frac{t_k}y)\frac{\bE[S_u(y)\mid t_k\le Y_u<y<Y_v]}{\bE[S_u(t_k)\mid Y_u<t_k<Y_v]}
\\ &\ge\frac{t_k}{y}\frac{\bE[S_u(t_k)\mid Y_u<t_k<Y_v]-1/T}{\bE[S_u(t_k)\mid Y_u<t_k<Y_v]}+(1-\frac{t_k}y)\frac{1-2y}{\bE[S_u(t_k)\mid Y_u<t_k<Y_v]}
\\ &\ge\frac{t_k}{y}(1-\frac{1}{CT})+(1-\frac{t_k}y)(1-2y)
\\ &=1-\frac{t_k}{yCT}-2(y-t_k) \ge1-\frac{3}{CT}
\end{align*}
We explain the first inequality.  The first term is because for $S_u(t_k)-S_u(y)=1$ (conditional on $Y_u<t_k$), we need another vertex to arrive in the interval $[t_k,y)$ and select $u$.  The sum of expectation of these events is at most $1/T$.  The second term is because for $S_u(y)$ to be 0, we need another vertex to arriving before $y$ and select $u$, or need $u$ to select another vertex that arrives before $y$.  The second inequality applies the lower bound of $C$ and the upper bound of $1$.
\end{proof}

Observe that \Cref{lem:continuous_induction_step} and \Cref{lem:discretization_step}
together complete the inductive step. Thus, \eqref{eqn:IHmatchingProb} holds for all $0 \le j \le T-1$. We now analyze the sampling used by \Cref{alg:recursive_rcrs}.
Note that if $\delta =0$ (and we ignore the issue of computational efficiency), then this part of the proof is unnecessary.

\begin{lemma}[Sample Complexity] \label{lem:sampling_step}
Given $\delta \in (0,1)$ and $T \in \mb{N}$, suppose that $Q \ge \frac{3}{C\delta^2}\log\frac{2Tn^2}{\delta}$ in \Cref{alg:recursive_rcrs}.
Then, $\mb{P}[ \cup_{j =0}^k B_j] \le \frac{\delta k}{T}$ for each $k=0, \ldots ,T-1$.
\end{lemma}

\begin{proof}[Proof of \Cref{lem:sampling_step}]
We proceed inductively. Clearly, if $k=0$, then each instantiation of $\hS(0)$ satisfies $\neg B_0$ by definition, and so the base
case holds. Let us now prove the inductive step. Fix $1 \le k \le T-1$, and assume that $\mb{P}[ \cup_{j =0}^{k-1} B_j] \le \frac{\delta (k-1)}{T}$. Now, by applying trivial upper bounds,
\begin{align*}
\mb{P}[ \cup_{j =0}^k B_j]  &= \mb{P}[B_k \mid  \cap_{j =0}^{k-1} \neg B_j] \cdot \mb{P}[\cap_{j =0}^{k-1} \neg B_j]  + \mb{P}[\cup_{j =0}^k B_j \mid  \cup_{j =0}^{k-1} B_j] \cdot \mb{P}[\cup_{j =0}^{k-1} B_j] \\
                            &\le \mb{P}[B_k \mid  \cap_{j =0}^{k-1} \neg B_j] + \mb{P}[\cup_{j =0}^{k-1} B_j] \\
                            &\le \mb{P}[B_k \mid  \cap_{j =0}^{k-1} \neg B_j] + \frac{\delta (k-1)}{T},
\end{align*}
where the last line follows from the induction hypothesis. It suffices to prove
that $\mb{P}[B_k \mid  \cap_{j =0}^{k-1} \neg B_j]  \le \delta/T$. Now, if $(\hs(i))_{i=0}^{k-1}$ satisfies $\cap_{j =0}^{k-1} \neg B_j$, then \Cref{lem:continuous_induction_step} implies
that $\bE[S_u(t_k) \mid Y_u<t_k<Y_v, (\hS(i))_{i=0}^{k-1} = (\hs(i))_{i=0}^{k-1}] \ge c(t_k)$ for all $u,v \in V$. On the other hand,
we know that $c(t_k) \ge C > 0$, as $c$ is a selection function. Thus, since
we take $Q \ge \frac{3}{C\delta^2}\log\frac{2Tn^2}{\delta}$ samples,
we can apply the multiplicative Chernoff bound to $\hS_{v \ra u }(k)$ and
its expectation $\mu = \bE[S_u(t_{k}) \mid Y_u<t_{j}<Y_v, (\hS(i))_{i=0}^{k-1} = (\hs(i))_{i=0}^{k-1}] \ge C$, to get that $$\mb{P}[ |\hS_{v \ra u }(k) - \mu| \ge \delta \mu \mid (\hS(i))_{i=0}^{k-1} = (\hs(i))_{i=0}^{k-1}] \le \frac{\delta}{Tn^2}.$$
This holds for each $(\hs(i))_{i=0}^{k-1}$ which satisfies $\cap_{j =0}^{k-1} \neg B_j$,
and so $\mb{P}[|\hS_{v \ra u }(k) - \mu| \ge \delta \mu \mid \cap_{j =0}^{k-1} \neg B_j] \le \frac{\delta}{Tn^2}.$ By union bounding over all $u,v \in V$, we get that $\mb{P}[B_k \mid  \cap_{j =0}^{k-1} \neg B_j]  \le \delta/T$,
which completes the proof.
\end{proof}

In \Cref{lem:sampling_step} we analyzed how the number of samples $Q$ should be set as a function of parameters $T$ and $\delta$.  We now finish the proof of \Cref{thm:recursive_rcrs} by specifying how $T$ and $\delta$ should be set given a desired error threshold $\eps$.

\begin{proof}[Proof of \Cref{thm:recursive_rcrs}]
Fix $\eps > 0$. Suppose that $\scr{M}$ is the matching returned by \Cref{alg:recursive_rcrs} when
executed with parameters $T,Q \in \mb{N}$ and $\delta \in (0,1)$. We shall prove
that if $Q \ge \frac{3}{C\delta^2}\log\frac{2Tn^2}{\delta}$, then for all $e \in E$, 
\begin{equation} \label{eqn:edge_guarantee_rough}
    \frac{\mb{P}[e \in \scr{M}]}{x_e} \ge \frac{(1-\delta)^2}{1+\delta} \left(1 -  \frac{8}{C^2T} \right)  \int_0^1 2 c(y) y dy.
\end{equation}
By taking $T \ge 8/C^2 \eps$ and $\delta \le \eps/(3-\eps)$, this implies
that $\frac{\mb{P}[ e \in \scr{M}]}{x_e} \ge (1 - \eps)^2 \int_0^1 2 c(y) y dy$ for all $e \in E$,
as desired.

Fix $u,v \in V$ with $e =(u,v) \in E$. Let us first assume that the instantiation $(\hs(i))_{i=0}^{T-1}$ of $(\hS(i))_{i=0}^{T-1}$ satisfies $\cap_{i=0}^{T-1} \neg B_i$. In this case,
for each $j=0, \ldots ,T -1$, \eqref{eqn:IHmatchingProb} implies that for all $y \in I_{j+1}$,
$$\bE[M_{v\to u}(y) \mid Y_u<Y_v, Y_v=y, F_v =u, (\hS(i))_{i=0}^{j}= (\hs(i))_{i=0}^{j}] \ge
c(y)\left(1-\frac{4}{CTy}\right).$$
Thus, since $\mb{P}[F_v =u] =x_{e}$ and $\mb{P}[Y_u < Y_v \mid Y_v = y] = y$, we can use the
independence of $(\hS(i))_{i=0}^{j}$ to get that
$$
\bE[M_{v\to u}(y) \mid Y_v=y, (\hS(i))_{i=0}^{j} = (\hs(i))_{i=0}^j] \ge  y x_e c(y)\left(1-\frac{4}{CTy}\right).
$$
As such, $\bE[M_{v\to u}(t_{j+1}) \mid (\hS(i))_{i=0}^{j}= (\hs(i))_{i=0}^{j}] - \bE[M_{v\to u}(t_{j}) \mid (\hS(i))_{i=0}^{j}= (\hs(i))_{i=0}^{j}]$ is lower bounded by
\begin{align*}
 \int_{t_j}^{t_{j+1}} \bE[M_{v\to u}(y) \mid Y_v=y, (\hS(i))_{i=0}^{j}= (\hs(i))_{i=0}^{j}] dy  &\ge x_e\frac{1-\delta}{1+\delta}\int_{t_j}^{t_{j+1}} c(y)\left(1-\frac{4}{CTy}\right) y dy \\
                                 &\ge  x_e \frac{1-\delta}{1+\delta}\left( \int_{t_j}^{t_{j+1}} c(y) y dy  -  \frac{4}{CT^2} \right), 
\end{align*}
where the last inequality uses $c(y) \le 1$, and $t_{j+1} -t_j =1/T$. Thus, since $\mb{P}[\cap_{i=0}^{j} \neg B_i] \ge 1- \delta$
by \Cref{lem:sampling_step}, we get that
$$\bE[M_{v\to u}(t_{j+1})] - \bE[M_{v\to u}(t_{j})] \ge  x_e \frac{(1-\delta)^2}{1+\delta}\left( \int_{t_j}^{t_{j+1}} c(y) y dy  -  \frac{4}{CT^2} \right).$$
Therefore, 
\begin{align*}
    \bE[M_{v\to u}(1)] &= \sum_{j=0}^{T-1} (\bE[M_{v\to u}(t_{j+1})] - \bE[M_{v\to u}(t_{j})]) \\
    &\ge x_e \frac{(1-\delta)^2}{1+\delta}  \sum_{j=0}^{T-1}\left( \int_{t_j}^{t_{j+1}} c(y) y dy  -  \frac{4}{CT^2} \right) = x_e \frac{(1-\delta)^2}{1+\delta}\left( \int_0^1 c(y) y dy  -  \frac{4}{CT} \right).
\end{align*}
The same lower bound holds for $\bE[M_{u\to v}(1)]$, and so
\begin{align*}
    \mb{P}[e \in \scr{M}]= \bE[M_{v\to u}(1)] + \bE[M_{u\to v}(1) ] &\ge  x_e \frac{(1-\delta)^2}{1+\delta}\left( \int_0^1 2 c(y) y dy  -  \frac{8}{CT} \right) \\
    &\ge x_e \frac{(1-\delta)^2}{1+\delta} \left(1 -  \frac{8}{C^2T} \right)  \int_0^1 2 c(y) y dy,
\end{align*}
where the second inequality uses that  $\int_0^1 2 c(y) y dy \ge  \int_0^1 2 C y dy= C$. 
Thus, \eqref{eqn:edge_guarantee_rough} holds, and so the proof is complete.
\end{proof}

\subsection{Proving \Cref{lem:continuous_induction_step}} \label{sec:proving_continuous_induction_step}

Let $k \ge 1$ and fix $(u,v)$.   Our goal is to prove that for any instantiation  $(\hs(i))_{i=0}^{k-1}$  of $(\hS(i))_{i=0}^{k-1}$ which satisfies
$\cap_{i=0}^{k-1} \neg B_i$, we have that $\mb{E}[ S_{u}(t_k) \mid Y_u < t_k < Y_v, (\hS(i))_{i=0}^{k-1} = (\hs(i))_{i=0}^{k-1}] \ge c(t_k)$. In other words, assuming the induction hypothesis upper bound on match probabilities before time $t_k$, we must prove a lower bound on $u$ being safe at time $t_k$. In order to simplify the equations, we implicitly condition on $(\hS(i))_{i=0}^{k-1} = (\hs(i))_{i=0}^{k-1}$ throughout the section.

Recall that $M_{u}(t_k) := 1 - S_{u}(t_k)$. I.e.,
$M_{u}(t_k)$ is an indicator random variable for the event that $u$ is matched by time $t_k$.
Our goal is to show that $\mb{E}[M_{u}(t_k) \mid Y_u < t_k < Y_v] \le 1 - c(t_k)$. 
We first show that if we remove the conditioning on $Y_v > t_k$,
then we can upper bound $\mb{E}[ M_{u}(t_k) \mid Y_u < t_k]$ using the induction hypothesis 
\eqref{eqn:IHmatchingProb} in terms of an integral of the selection function $c$.

First recall that $M_{u}(t_k) = \sum_{w \in N(u)} (M_{w \rightarrow u}(t_k) + M_{u \rightarrow w}(t_k))$. Thus, after taking conditional expectations,
\begin{equation} \label{eqn:unconditioned_sum}
    \mb{E}[M_{u}(t_k) \mid Y_u < t_k] = \sum_{w } \mb{E}[ M_{w \rightarrow u}(t_k) 
    \mid Y_u < t_k] + \sum_{w} \mb{E}[M_{u \rightarrow w}(t_k) \mid Y_u < t_k].
\end{equation}
On the other hand, for each $w \in N(u)$, since $\mb{P}[Y_u < y_w \mid Y_w =y_w] = y_w$
and $\mb{P}[F_w = u] = x_{u,w}$,
\begin{equation*}
\mb{E}[ M_{w \ra u}(t_k)] = \int_{0}^{t_k}  x_{w,u} y_w \bE[M_{w \ra u}(y) \mid Y_u < y_w, Y_w =y_w,F_w=u] dy_w \le x_{u,w} \int_{0}^{t_k} y_w c(y_w) dy_w,
\end{equation*}
where the inequality follows from the upper bound of \eqref{eqn:IHmatchingProb}. Thus, 
since $Y_u < t_k$ is a necessary condition for $M_{w \ra u}(t_k) =1$, we have that
$\mb{E}[ M_{w \ra u}(t_k) \mid Y_u < t_k] \le \frac{1}{t_k} \int_{0}^{t_k} y_w c(y_w) dy_w.$
The same upper bound holds for $\mb{E}[ M_{u \ra w}(t_k) \mid Y_u < t_k]$, and so applied to \eqref{eqn:unconditioned_sum},
we get that
\begin{equation} \label{eqn:unconditioned_integral}
    \mb{E}[M_{u}(t_k) \mid Y_u < t_k] \le \sum_{w \in N(u)} 2 x_{w,u} \int_{0}^{t_k} y_w c(y_w) dy_w \le \int_{0}^{t_k} 2 c(y) y dy,
\end{equation}
where the final inequality uses $\sum_{w \in N(u)} x_{w,u} \le 1$.

We now upper bound the \textit{positive correlation} between the events $Y_v > t_k$
and $M_{u}(t_k) =1$. Specifically, we
prove that if $\phi_{g}(y):=\frac{2y^{g-1}}{(g-1)!}$, where $\phi_{\infty}(y) := 0$,
then
\begin{equation} \label{eqn:difference_in_conditioning}
    \mb{E}[ M_{u}(t_k) \mid Y_u < t_k < Y_v] - \mb{E}[ M_{u}(t_k) \mid Y_u < t_k] \le \frac{1}{t_k} \int_{0}^{t_k} \phi_{g}(y) dy
\end{equation}
This will complete the proof of \Cref{lem:continuous_induction_step}, as \eqref{eqn:unconditioned_integral}
and \eqref{eqn:difference_in_conditioning} imply that
\begin{align*}
\mb{E}[M_{u}(t_k) \mid Y_u < t_k < Y_v] &= \mb{E}[M_{u}(t_k) \mid Y_u < t_k] + \mb{E}[ M_{u}(t_k) \mid Y_u < t_k < Y_v] - \mb{E}[ M_{u}(t_k) \mid Y_u < t_k]
\\ &\le\frac1{t_k}\int_0^{t_k}2 c(y) y dy + \frac{1}{t_k} \int_{0}^{t_k} \phi_{g}(y) dy
\\ &=\frac1{t_k}\int_0^{t_k}2(c(y) y+ \phi_{g}(y))dy \\
    &\le 1 - c(t_k),
\end{align*}
where the last inequality uses that $c$ is a selection function. Thus, we spend the remainder
of the section verifying \eqref{eqn:difference_in_conditioning}.

In order to prove \eqref{eqn:difference_in_conditioning}, we couple two executions of \Cref{alg:recursive_rcrs} up until time $t_k$. The first is the \textit{usual} execution of \Cref{alg:recursive_rcrs} on $G$ with the random variables $(Y_w)_{w \in V}$, $(F_w)_{w \in V}$,
    and $(A_{w_2 \ra w_1})_{w_1,w_2}$. The second is a \textit{parallel} execution on $G^{-v} := G \setminus \{v\}$ and thus does \textit{not} include the vertex $v$. Our coupling uses the same random variables $(Y_w)_{w \neq v}$, $(F_{w})_{w \neq v}$, and $(A_{w_2 \ra w_1})_{w_1,w_2 \neq v}$ as in the first execution. However, we now define $M^{-v}_{w \ra u}(t_k)$ (respectively, $M^{-v}(u)$) to be the indicator random variable for the event $w$ selects $u$ (respectively, $u$ is matched) before time $t_k$ when executing on $G^{-v}$. 
Observe first that by definition $\bE[M_u(t_k) \mid Y_u<t_k<Y_v] = \bE[M^{-v}_u(t_k) \mid Y_u < t_k]$,
and so it suffices to upper bound $\bE[M^{-v}_u(t_k)-M_u(t_k) \mid Y_u<t_k]$ when proving \eqref{eqn:difference_in_conditioning}. In order to do so, we define a ``bad'' event $B_{v \ra u}$
such that
\begin{equation} \label{eqn:indicator_bad_event}
    M^{-v}_u(t_k)-M_u(t_k) \le \bm{1}_{B_{v \ra u}}
\end{equation}
Afterwards, we prove that $\mb{P}[B_{v \ra u} \mid Y_u < t_k] \le \int_{0}^{t_k} \phi_{g}(y) dy$,
which together with the previous equation implies \eqref{eqn:difference_in_conditioning}.

Before defining $B_{v \ra u}$, we first show that a large part of our argument holds
even if $B_{v \ra u} = \emptyset$ (and so $\bm{1}_{B_{v \ra u}}$ in \eqref{eqn:indicator_bad_event} is identically $0$). Let us first write
$M^{-v}_u(t_k) = \sum_{w \in N(u) \setminus \{v\}}M^{-v}_{w\to u}(t_k) + \sum_{w \in N(u) \setminus \{v\}}M^{-v}_{u \ra w}(t_k)$. Observe that at most one of the sums is non-zero, and so we can handle each separately. We begin with the left-most sum, and show that if a vertex $w \in N(u) \setminus \{v\}$ selects $u$ during the
execution on $G^{-v}$, then $u$ must be matched during the execution on $G$. 
\begin{lemma}\label{prop:vertex_selected_unconditional}$
\sum_{w \in N(u) \setminus \{v\}}M^{-v}_{w\to u}(t_k)
\le M_{u}(t_k).$
\end{lemma}

\begin{proof}

Set $M^{-v}_{ \boldsymbol{\cdot} \to u} (t_k):= \sum_{w \neq v} M^{-v}_{w\to u}(t_k)$ for convenience. Clearly, $M^{-v}_{ \boldsymbol{\cdot} \to u}(t_k) , M_{u}(t_k) \in \{0,1\}$. Observe that in order to prove the proposition, it suffices to argue that $M^{-v}_{ \boldsymbol{\cdot} \to u}(t_k) \le M_{u}(t_k)$. Now,
    if $Y_u > t_k$, then $u$ cannot be selected before time $t_k$, so $M^{-v}_{ \boldsymbol{\cdot} \to u}(t_k) = 0$. Similarly, if $Y_v > t_k$, then the executions are the same, and so $M^{-v}_{ \boldsymbol{\cdot} \to u}(t_k) = M_{u}(t_k)$.  Thus, it suffices to consider the case when $Y_v, Y_u < t_k$.

    Let us first consider the execution on $G^{-v}$, and assume that $w \neq v$ satisfies $M^{-v}_{w \ra u}(t_k)=1$. Observe that
    since $M^{-v}_{w \ra u}(t_k) =1$, $u$ did \textit{not} select a vertex when it arrived. Moreover, $w$ is the first vertex with $Y_u<Y_{w} < t_k$, $F_{w} = u$,
    and $A_{w \ra u} =1$. We now compare the execution on $G$. Clearly, we may assume that $u$ did not select a vertex when it arrived, as otherwise $M_{u}(t_k) =1$. But now $u$ is selected if and only if there exists a vertex $w'$ with $Y_u<Y_{w'} < t_k$, $F_{w'} = u$ and $A_{w' \ra u}=1$. Clearly, $w$ satisfies these properties, and so the proof is complete.
\end{proof}

The only remaining case is to handle when $\sum_{w \in N(u) \setminus \{v\}}M^{-v}_{u \ra w}(t_k) = 1$, and yet $M_{u}(t_k)=0$. Roughly speaking, we show that if both these events occurs, then there must exist a path of \textit{odd} length between $v$ and $F_u$. Now, if $G$ were bipartite, then such a path cannot exist, as it would imply an odd length cycle including $u$ and $v$ exists. If $G$ is not bipartite, then there may be an odd length path between $v$ and $F_u$. In this case,
we argue that such a path must satisfy certain additional properties, which we then
argue are unlikely to be satisfied. 

Given an edge $(w_1,w_2) \in E$, we say that $(w_1,w_2)$ \textit{survives}, provided
$Y_{w_1} < Y_{w_2}$, $F_{w_2} = w_1$, and $A_{w_2 \ra w_1} =1$, or $Y_{w_2} < Y_{w_1}$, $F_{w_1} = w_2$, and $A_{w_1 \ra w_2} =1$. We define its arrival time $Y_{(w_1,w_2)}$ to be $\max\{Y_{w_1},Y_{w_2}\}$. Suppose that $P =v_1, \ldots ,v_d$ is a path with $v_1 = v$ of odd length (i.e., $d$ is even).
We say that $P$ is a \textit{flipping sequence} for $u$ via $v$, provided the following events occur:
\begin{enumerate}
    \item $F_u = v_d$;
    \item $Y_{v_i} < Y_u$ for all $i=1, \ldots ,d$;
    \item The edges $\{(v_{i},v_{i+1})\}_{i=1}^{d-1}$ all survive, and $Y_{(v_1,v_2)} <  \ldots < Y_{(v_{d-1},v_d)}$.
\end{enumerate}
We define the event $B_{v \ra u}$ as there existing a flipping sequence in $G$
for $u$ via $v$. (This event is \textit{not} symmetric in $u$ and $v$.)
The intuition is that such a sequence is necessary for the exclusion of $v$ to "flip" the matched status of $u$ from 0 to 1.
See the later \Cref{fig:pentagonEg} for an illustration.

We now prove that if $\sum_{w \in N(u) \setminus \{v\}}M^{-v}_{u \ra w}(t_k) = 1$ and $M_{u}(t_k)=0$, then $B_{v \ra u}$ occurs (i.e.\ there exists a flipping sequence). The idea is
to show that if we exclude $v$, then the difference between the sets of vertices matched in the two executions, denoted $V(\scr{M}) \setminus \{v\}$ and $V(\scr{M}^{-v})$, is at most one \textit{critical vertex} $v_{c}$ at all times. The critical vertex begins as neighbor of $v$, and initially is in $V(\scr{M})\setminus\{v\}$. When a new edge $(v_c, w)$ is processed and matched in one of the executions but not the other, $w$ becomes the new critical vertex. Observe that if $v_c$ was previously in $V(\scr{M}) \setminus \{v\}$, then $w$ will be in $V(\scr{M}^{-v})$ (and vice versa). Since $M_{u}(t_k)=0$ in the usual execution but $M^{-v}_{u\to w}(t_k)=1$ for some $w = F_u \in N(u)\setminus\{v\}$ in the $v$-excluded execution, if we prepend $v$ to the path formed
by the edges that switch the critical vertex before time $Y_u$, we get
a flipping sequence. We formalize this argument in the following \namecref{prop:flippingSeq}.

\begin{lemma} \label{prop:flippingSeq}
    If $\sum_{w \in N(u) \setminus \{v\}}M^{-v}_{u \ra w}(t_k) =1$ and $M_{u}(t_k)=0$,
    then $B_{v \ra u}$ occurs.
\end{lemma}
\begin{proof}

First observe that since $\sum_{w \in N(u) \setminus \{v\}}M^{-v}_{u \ra w}(t_k) =1$,
we know that $F_u \neq \perp$.
It suffices to analyze both executions up until time $Y_u$ when the edge $(u,F_u)$ is considered. Now, if $v$ is not matched
before time $Y_u$ in the first execution, then the matchings output by both executions are the same at time
$Y_u$. Thus, since $\sum_{w \in N(u) \setminus \{v\}}M^{-v}_{u \ra w}(t_k) =1$ and $M_{u}(t_k)=0$ by assumption, we know that know that there must exist some edge $e_1 =(w_1, v) \in N(v)$ which matches to $v$ in the original execution. Let us refer to an edge $f \in E$ as \textit{relevant},
provided
\begin{enumerate}
    \item $f$ survives
    \item $Y_{e_1} \le Y_f \le Y_u$
\end{enumerate}
Clearly, $e_1$ is the first relevant edge, and $(u,F_u)$ is the last relevant edge. Let us order the relevant edges in terms by their arrival times: i.e.,  $e_1, e_2, \ldots ,e_l$, where $\ell \ge 2$. We denote $\scr{M}_i$ (respectively, $\scr{M}^{-v}_i$) to be the edges matched \textit{after} the first $i$ relevant edges are processed in the
usual (respectively, parallel) execution. If a relevant edge is incident to $v$, the parallel execution passes on the edge.

Now, $\scr{M}_0 = \scr{M}^{-v}_0$, as the executions
match the same edges before $v$ is matched. We claim that
$V(\scr{M}^{v}_i)$ and $V(\scr{M}_i) \setminus \{v\}$ differ by exactly $1$ vertex for each $i=1, \ldots , \ell$.
When $i=1$, this is clearly true, as $\scr{M}^{-v}_1 = \scr{M}^{-v}_0$, and $\scr{M}_1 = \scr{M}_0 \cup \{e_1\}$. In general,
let us proceed inductively and assume this is true for $i \ge 1$. We denote the differing vertex by $v_{c}(i)$, which we refer to as the \textit{critical vertex} after $i$ steps. Now consider when edge $e_{i+1}$ is processed. First observe
that if $e_{i+1}$ is not incident to $v_{c}(i)$, then the executions either both match $e_{i+1}$ or both don't match $e_{i+1}$,
so the claim still holds. Alternatively, if $e_{i+1}=(w,v_c(i))$ for some $w \in V$, then the execution which \textit{hadn't} matched $v_{c}(i)$ will add the edge $e_{i+1}$, whereas the execution which \textit{had} matched
$v_{c}(i)$ will not. In this case, the claim still holds, and the new critical vertex is $w$.

Observe that $v_{c}(1) = w_1$, where $e_1 = (w_1, v)$ is the first relevant edge, $v_{c}(\ell) = u$,
and $v_{c}(\ell-1) = F_u$. Moreover, for each $i=2, \ldots , \ell$, either $v_{c}(i) = v_{c}(i-1)$, or $e_i=(v_{c}(i-1), v_{c}(i))$. We shall construct a flipping sequence using $v_{c}(1) \ldots , v_{c}(\ell-1)$ and the edges between them, however we first need to remove the (potential) repetitions. Observe that for each $i \ge2$, if $v_{c}(i) \neq v_{c}(i-1)$,
then it is distinct from all of $v_{c}(1), \ldots , v_{c}(i-1)$. This is because if the critical vertex changed in step $i$, then $e_i$ must have been matched. Since $v_{c}(1), \ldots , v_{c}(i-1)$ were previously matched by both executions, $e_i$ could not have been incident to any of them, and so they couldn't have become the critical vertex $v_{c}(i)$.
Due to this property, if we remove the repetitions from $v_{c}(1), \ldots , v_{c}(\ell-1)$, we're left with a path $P'$
of distinct vertices which begins at $w_1$ and ends at $F_u$. By prepending the vertex $v$, we claim that
$v P'$ is a flipping sequence. Clearly, $v P'$ is between $v$ and $F_u$, and contains surviving edges which arrive in increasing order before time $Y_u$, so it remains to argue $|V(vP')|$ is even. To see this, notice that since we've removed the repetitions, the critical vertex alternatives between $V(\scr{M}_j) \setminus \{v\}$ and $V(\scr{M}^{-v}_j)$ in each sub-step $j$. Since it begins in $V(\scr{M}_1)$ and ends in $V(\scr{M}_{\ell-1})$, $|V(P')|$ is odd, and so $|V(v P')|$ is even.
The proof is thus complete.
\end{proof}

It remains to prove $\mb{P}[B_{v \ra u} \mid  Y_u < t_k] \le \int_{0}^{t_k} \phi_{g}(y) dy$. 
Now, if $G$ is bipartite, then there is no path of odd-length between $u$ and $F_u$, so $\mb{P}[B_{v \ra u} \mid  Y_u < t_k] =0$ as desired. Let us thus instead assume that $G$ is not bipartite and has odd-girth $g \ge 3$.

\begin{figure}
\centering
\begin{tikzpicture}[>=Stealth, every node/.style={circle, draw, minimum size=0.8cm, inner sep=0pt}]

\foreach \i/\label in {1/$v_2$, 2/$v$, 3/$u$, 4/$v_4$, 5/$v_3$} {
\node (\i) at ({72*(\i-1)}:1.5) {\label};
}

\draw[->,dashed] (1) to[out=90, in=0] (2);
\draw[->,dashed] (2) to[out=270, in=180] (1);

\draw[->] (3) -- (4);
\draw[->] (4) -- (5);

\draw[->] (5) -- (1);

\end{tikzpicture}
\caption{
Example with odd girth 5, in which case $d=4$.  We are analyzing the probability that $v$ selects $u$ (at time $t_k$), which requires upper-bounding the probability that $u$ was already matched (before time $t_k$). $B_{v\to u}$ occurs when $u$ gets matched \textit{only in the parallel execution} where $v$ is "removed from the picture" (because we have conditioned on $v$ arriving at time $t_k$).  This requires a flipping sequence, which means that edges $(v,v_2),(v_2,v_3),(v_3,v_4)$ all survive and arrive in that order, all before $u$ arrives.
A pre-requisite for this occurring is path $P=v,v_2,v_3,v_4$ having potential, which requires all of the choices indicated in the figure by solid arrows (e.g., the arrow from $u$ to $v_4$ indicates that $F_u=v_4$) and at least one of the dashed choices.  Even an adversarial construction of $(F_w)_{w \in V}$ results in at most one path with potential. Conditional on $Y_u=y$, a path with potential becomes a flipping sequence with probability upper-bounded by $2y^4/4!$, because this requires all four vertices $v,v_2,v_3,v_4$ to arrive before time $y$, in order either $v,v_2,v_3,v_4$ or $v_2,v,v_3,v_4$ (this upper bound is tight if \textit{both} dashed choices occur, i.e.\ $F_v=v_2$ and $F_{v_2}=v$).}
\label{fig:pentagonEg}
\end{figure}
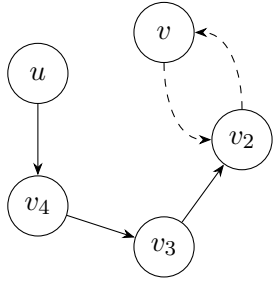

We make use of the decoupling between the random variables $(F_w)_{w \in V}$ and the arrival
times $(Y_w)_{w \in V}$. First, suppose that $P=v_1, \ldots ,v_d$ is an odd length path in $G$
with $v_1 = v$. In order for $P$ to have a chance at being a flipping sequence, its edges must be realized
using $(F_w)_{w \in V}$ in a very specific way. We say that $P$ has \textit{potential} provided the following
events occur:
\begin{enumerate}
    \item $F_{u} = v_d$;
    \item $F_{v_1} = v_2$ or $F_{v_2} = v_1$;
    \item $F_{v_i} = v_{i-1}$ for all $i=3, \ldots ,d$;
\end{enumerate}
See \Cref{fig:pentagonEg} for an illustration of a path with potential.  Observe that these events only depend on the $(F_w)_{w \in V}$ random variables.
Since the graph on $V$ with edge set $\{ (w,F_w): w \in V, F_w \neq \perp\}$ has
at \textit{most} $|V|$ edges, the next \namecref{eqn:two_flipping} follows immediately.

\begin{proposition} \label{eqn:two_flipping}
There is at most one path in $G$ with potential.
\end{proposition}
We now consider the arrival times of an odd-length path $P=v_1, \ldots ,v_{d}$ with $v_1 = v$. We say that $P$ is \textit{badly ordered}, provided
\begin{enumerate}
    \item $Y_{v_i} < Y_u$ for all $i=1, \ldots ,d$;
    \item $Y_{v_1} < \ldots  < Y_{v_d}$ or $Y_{v_2} < Y_{v_1} < Y_{v_3} < \ldots < Y_{v_d}$.
\end{enumerate}
The existence of a badly ordered path with potential is a necessary for there
to exist a flipping sequence of $u$ via $v$.  The following is immediate from definition.
\begin{proposition} \label{eqn:flipping_hazard_equivalence}
If $B_{v \ra u}$ occurs, then there exists a badly ordered path in $G$ with potential.
\end{proposition}
We now compute $\mb{P}[\text{$P=v_1, \ldots ,v_{d}$ is badly ordered} \mid Y_u = y]$
where $y \in [0,t_k)$.
Observe first that the probability $P$ is badly ordered is $2y^{d}/d!$, because we need all of $v_1,\ldots,v_{d}$ to arrive before time $y$, and conditional on this occurring, there $d!$ relative orderings, $2$ of which are permissible: $Y_{v_1} < \ldots  < Y_{v_d}$ and $Y_{v_2} < Y_{v_1} < Y_{v_3} < \ldots < Y_{v_d}$.
Thus, $\mb{P}[\text{$P=v_1, \ldots ,v_{d}$ is badly ordered}] = \int_{0}^{t_k} \frac{2 y^{d}}{d!} dy$. By applying \eqref{eqn:two_flipping} and \eqref{eqn:flipping_hazard_equivalence},
and using that any even-length path between $F_u$ and $v$ has length at least $g -1$, we get \eqref{eqn:difference_in_conditioning}, and so \Cref{lem:continuous_induction_step} is proven.  We note that via examples analogous to \Cref{fig:pentagonEg}, it is easy to see that our upper bound of $\int_0^{t_k} \frac{2y^d}{d!}dy$ is tight for any odd girth $g$ and $d=g-1$, once we concede to an adversarial choice of $(F_w)_{w\in V}$.

\section{Proving \Cref{thm:general_graph_positive}} \label{sec:general_graph_positive}

Let $\bm{x}=(x_e)_{e \in E}$ be a fractional matching in graph $G$.
For each time $t \in [0,1]$, we define a ``two-phased'' RCRS, where phase $1$ processes vertices $v$
with $Y_v < t$, and phase $2$ processes vertices $v$ with  $Y_v \ge t$. 
In order to describe each phase, we first define a function $a_{t}: [0,1] \rightarrow \mb{R}$, where 
\begin{equation}
    a_{t}(x):=\frac{3 + 6 t + 4 t^2 + 2 t^3}{3 + 6 t + 4 t^2 + 2 t^3 + 2x(1-t)(1+3t+t^2).}
\end{equation}
Note that $a_t(x) \in [0,1]$ for all $x \in [0,1]$.
We use function $a_t$ to ``prune'' each edge $e$ of $G$, based on its $x_e$ value. Specifically, we draw $A_e \sim \Ber(a_t(x_e))$ independently, and say that $e$ \textit{survives} if $e$ is active
and $A_e =1$. Observe that each $e$ survives with probability $f_t(x_e):=x_e a_t(x_e)$. Our RCRS will
select an edge only if the edge survives. Since $t$ is an algorithm parameter treated as fixed, we will sometimes drop the dependence on $t$ in our analysis (i.e., $f = f_t$ and $a =a_t$).

 In the first phase, when a vertex arrives, we execute the OCRS of \cite{Ezra_2020} on the ``pruned instance'' $(G, f(\bm{x}))$. This phase ensures that each
edge which arrives before time $t$ is selected with probability exactly $f(x_e)/2$. In the second phase, we execute the greedy RCRS on
$(G, f(\bm{x}))$. Specifically, when $e$
arrives, if the edge survives, then we add it to the matching
(if possible). Below is a formal description of our RCRS.

\begin{algorithm}[H]
\caption{Two-phase RCRS} 
\label{alg:two_phase}
\begin{algorithmic}[1]
\Require $G=(V,E)$, $\bm{x}=(x_e)_{e \in E}$, and $t \in [0,1]$
\Ensure subset of active edges forming a matching $\scr{M}$
\State $\scr{M} \leftarrow \emptyset$
\For{arriving vertices $v$ in increasing order of $Y_v$}
\If{$F_v = u$, and $Y_u < Y_v$} \Comment{$(u,v)$ is active}
\State Draw $A_{u,v} \sim \Ber(a_t(x_{u,v}))$ independently.
\If{$Y_v < t$}          \Comment{Phase $1$}
\State Set $b_{u,v}:= 1/(2 - \sum_{w < v} f_t(x_{u,w}))$.
\State Draw $B_{u,v} \sim \Ber(b_{u,v})$ independently.
\ElsIf{$Y_v \ge t$}         \Comment{Phase $2$}
\State Set $B_{u,v} :=1$.
\EndIf
\If{$u$ is unmatched,  and $A_{u,v} \cdot B_{u,v}=1$}
\State $\scr{M} \leftarrow \scr{M} \cup \{(u,v)\}$

\EndIf

\EndIf

\EndFor
\State \Return $\scr{M}$
\end{algorithmic}
\end{algorithm}
\begin{remark}
If $t = 0$, then $a_0(x)=3/(3+2x)$ which recovers the function of \cite{Fu2021}, so our algorithm is identical to their prune-greedy RCRS. If $t =1$,
then $a_t$ is identically $1$, and so we recover the OCRS of \cite{Ezra_2020}.  Also note that $a_t(0)=1$ regardless of the value of $t$.
\end{remark}
Suppose that $(u_0,u_1) \in E$ is an edge of $G$. We use the same indicator random variables as in the analysis of \Cref{alg:recursive_rcrs}. Specifically, $M_{u_0 \ra u_1}(y_0)=1$ (respectively, $M_{u_1}(y_0)=1$)
if $u_0$ selects $u_1$ (respectively, $u_1$ is matched) before time $y_0 \in [0,1]$. Observe first that
\begin{equation} \label{eqn:overall_guarantee}
\mb{P}[(u_0,u_1) \in \scr{M}] =  \int_{0}^{1} \mb{E}[M_{u_0 \ra u_1}(y_0)  \mid Y_{u_0} = y_0] dy_0 +
\int_{0}^{1} \mb{E}[M_{u_1 \ra u_0}(y_1)  \mid Y_{v_1} = y_1] dy_1.
\end{equation}
If $y_0 \le t$, then we use
the guarantee of the OCRS of \cite{Ezra_2020} applied to the pruned instance $(G, f(\bm{x}))$. This guarantee
holds for the edge $(u_0,u_1)$ no matter which subset of vertex arrival times $(Y_s)_{s \in S}$ for $S \subseteq V \setminus \{u_0,u_1\}$ we condition on:
\begin{lemma} \label{lem:ezra_guarantee}
Fix $(u_0,u_1) \in E$, $S \subseteq V \setminus \{u_0,u_1\}$ and $y_0 \le t$. Then 
$$\mb{E}[ M_{u_0 \ra u_1}(y_0) \mid (Y_s)_{s \in S}, Y_{u_0} = y_0, Y_{u_1} < y_0] = f(x_{u_0,u_1})/2.$$
\end{lemma}
Observe that by taking $S= \emptyset$ in \Cref{lem:ezra_guarantee},
$
\mb{E}[M_{u_0 \ra u_1}(y_0) \mid Y_{u_0} = y_0] = f(x_{u_0,u_1}) y_0/2
$ for $y_0 \in [0,t]$,
as $\mb{P}[Y_{u_1} < y_0] = y_0$. Applied to \eqref{eqn:overall_guarantee},
this means that
\begin{equation}\label{eqn:overall_guarantee_ezra}
\mb{P}[(u_0,u_1) \in \scr{M}] = \frac{f(x_{u_0,u_1})t^2}{2} + \int_{t}^{1} \mb{E}[M_{u_0 \ra u_1}(y_0)  \mid Y_{u_0} = y_0] dy_0 +
\int_{t}^{1} \mb{E}[M_{u_1 \ra u_0}(y_1)  \mid Y_{v_1} = y_1] dy_1.
\end{equation}
Now, the integrands are symmetric, so we focus on $\mb{E}[M_{u_0 \ra u_1}(y_0)  \mid Y_{u_0} = y_0]/f(x_0)$ for $y_0 \in (t,1]$.
We introduce a parameter $\ell \ge 1$, which indicates the graph distance from $u_0$ we consider
when bounding $\mb{E}[M_{u_0 \ra u_1}(y_0) \mid Y_{u_0} = y_0]/f(x_{u_0,u_1})$. Our bounds are defined recursively in terms
of $\ell$, and depend on the graph $G$, $y_0 \in [t,1]$, and the edge $(u_0,u_1)$ of $G$. While we ultimately require a lower bound on
$\mb{E}[M_{u_0 \ra u_1}(y_0)  \mid Y_{u_0} = y_0]/f(x_0)$, due to the recursive structure of our bounds, it will be convenient to simultaneously derive upper bounds.
Specifically, if $\ell$ is odd, we derive an upper bound, denoted
$\UB_{u_0 \ra u_1}^{G}(y_0, \ell)$. Otherwise, if $\ell$ is even, we derive a lower bound, denoted $\LB_{u_0 \ra u_1}^{G}(y_0,\ell)$.  As we explain in \Cref{sec:iterate_recursion}, the analysis of
the prune-greedy RCRS of \cite{Fu2021} can be recovered
by taking $\ell=4$ and setting $t=0$.
We now define our bounds and relate
them to $\mb{E}[M_{u_0 \ra u_1}(y_0) \mid Y_{u_0} = y_0]/f(x_{u_0,u_1})$
in \Cref{lem:recursion_two_phase}.

\textbf{Base case:} For $\ell =1$, set $\UB_{u_0 \ra u_1}^{G}(y_0,1):=y_0$

\textbf{Recursion:}
Given $\ell \ge 2$ even,  assume we are given values $\UB_{u_1 \ra u_2}^{G \setminus u_0}(y, \ell-1)$
and $\UB_{u_2 \ra u_1}^{G \setminus u_0}(y, \ell-1)$ for each $u_2 \in N(u_1) \setminus u_0$ and $y \in (t,1]$. In this case, define
\begin{equation*}
    \LB_{u_0 \ra u_1}^{G}(y_0,\ell) :=  y_0  - \sum_{u_2  \in N(u_1) \setminus u_0} 
 f(x_{u_1,u_2}) \left(\frac{t^2}{2} +  \int_{t}^{y_0} \UB^{G \setminus u_0}_{u_1 \ra u_2}(y_1, \ell-1)  dy_1 + \int_{t}^{y_0} \UB^{G \setminus u_0}_{u_2 \ra u_1}(y_2, \ell-1)  dy_2\right).
\end{equation*}
Otherwise, if $\ell \ge 3$ is odd, assume we are given values $\LB_{u_1 \ra u_2}^{G \setminus u_0}(y_1, \ell-1)$
and $\LB_{u_2 \ra u_1}^{G \setminus u_0}(y, \ell-1)$ for each $u_2 \in N(u_1) \setminus u_0$ and $y \in (t,1]$. Then, define
\begin{equation*}
    \UB_{u_0 \ra u_1}^{G}(y_0,\ell) :=  y_0  - \sum_{u_2  \in N(u_1) \setminus u_0} f(x_{u_1,u_2})  \left(\frac{t^2}{2} +  \int_{t}^{y_0} \LB^{G \setminus u_0}_{u_1 \ra u_2}(y_1, \ell-1)  dy_1 + \int_{t}^{y_0} \LB^{G \setminus u_0}_{u_2 \ra u_1}(y_2, \ell-1)  dy_2\right).
\end{equation*}

\begin{lemma} \label{lem:recursion_two_phase}
Fix $\ell \ge 1$. For any graph $G=(V,E)$, $(u_0,u_1) \in E$  and $y_0 \in (t,1]$:
\begin{itemize}
    \item $\mb{E}[ M_{u_0 \ra u_1}(y_0) \mid Y_{u_0} = y_0] \le f(x_{u_0,u_1}) \cdot \UB_{u_0 \ra u_1}^{G}(y_0, \ell)$ if $\ell$ is odd. 
    \item $\LB_{u_0 \ra u_1}^{G}(y_0, \ell) \cdot f(x_{u_0,u_1}) \le \mb{E}[ M_{u_0 \ra u_1}(y_0) \mid Y_{u_0} = y_0]$ if $\ell$ is even.
\end{itemize}
\end{lemma}
\begin{proof}
When $\ell = 1$, the lemma follows easily:
For each $(u_0,u_1) \in E$ of $G$, if $M_{u_0 \ra u_1}(y_0) =1$ and $Y_{u_0} = y_0$,
then $Y_{u_1} < y_0$ and $(u_0,u_1)$ survives.
The latter two events occur with probability $f(x_{u_0,u_1}) y_0$,
so $\mb{E}[M_{u_0 \ra u_1}(y_0) \mid Y_{u_0} =y_0] \le f(x_{u_0,u_1}) \cdot \UB_{u_0 \ra u_1}^{G}(y_0, 1)$.

Fix $\ell \ge 2$. We proceed inductively,
where we assume the lemma holds for $\ell-1$, and prove
that it holds for $\ell$. Let us assume that $\ell$ is even,
as the case when $\ell$ is odd is almost identical. 
To work with our induction hypothesis, it is
convenient to introduce additional indicator random variables for the execution
of \Cref{alg:two_phase} on $G \setminus u_0$. Specifically, for each $y \in (t,y_0)$ and $u_1, u_2 \in V \setminus u_0$, define $M^{-u_0}_{u_1 \ra u_2}(y)$ to be the indicator random variable for the event $u_1$ selects $u_2$
by time $y$ when executing \Cref{alg:two_phase} on $G \setminus u_0$. Our induction hypothesis implies that for each $y \in (t, y_0)$ and $u_1,u_2 \in V \setminus u_0$
\begin{equation} \label{eqn:induction_two_phase}
    \mb{E}[M^{-u_0}_{u_1 \ra u_2}(y) \mid  Y_{u_1} =y] \le f(x_{u_1,u_2}) \cdot \UB_{u_1 \ra u_2}^{G \setminus u_0}(y, \ell-1).
\end{equation}
Our goal is now to prove that $\mb{E}[M_{u_0 \ra v_1}(y_0) \mid Y_{u_0} = y_0] \ge \LB_{u_0 \ra u_1}^{G}(y_0, \ell) \cdot f(x_{u_0,u_1})$. Observe first that if we condition on $Y_{u_0} = y_{u_0}$,
then $M_{u_0 \ra v_1}(y_0)$ occurs if and only if $Y_{u_1} < y_0$, 
$(u_0,u_1)$ survives, and $M_{u_1}(y_0) = 0$. Thus,
\begin{align*}
\mb{E}[M_{u_0 \ra u_1}(y_0)  \mid Y_{u_0} = y_{0}] &= f(x_{u_0,u_1})( \mb{P}[Y_{u_1} < y_0 \mid Y_{u_0} =y_0] - \E[M_{u_1}(y_0) \mid Y_{u_0} =y_0]) \\
                                &=f(x_{u_0,u_1}) ( y_0 - \E[M_{u_1}(y_{0}) \mid Y_{u_0} =y_0]),
\end{align*}
Now, $$\mb{E}[M_{u_1}(y_0) \mid Y_{u_0} = y_0] = \sum_{u_2 \in N(u_1) \setminus u_0} \left( \mb{E}[M_{u_1 \ra u_2}(y_0) \mid Y_{u_0} = y_0] + \mb{E}[M_{u_2 \ra u_1}(y_0) \mid Y_{u_0} = y_0] \right).$$
We focus on the term $\mb{E}[M_{u_1 \ra u_2}(y_0) \mid Y_{u_0} = y_0]$, as $\mb{E}[M_{u_2 \ra u_1}(y_0) \mid Y_{u_0} = y_0]$ is symmetric.
By definition,
$$
\mb{E}[M_{u_1 \ra u_2}(y_0) \mid Y_{u_0} = y_0] = \int_{0}^{y_0} \mb{E}[M_{u_1 \ra u_2}(y_1) \mid Y_{u_0} = y_0, Y_{u_1} =y_1] dy_1.
$$
By applying \Cref{lem:ezra_guarantee} with $S= \{u_0\}$,
$
\mb{E}[M_{u_1 \ra u_2}(y_1) \mid Y_{u_0} = y_0, Y_{u_1} = y_1, Y_{u_2} < y_1]= f(x_{u_1,u_2})/2,
$
and so since $\mb{P}[Y_{u_2} < y_1 \mid Y_{u_0} = y_0, Y_{u_1} = y_1] =y_1$, 
$$
\int_{0}^{t} \mb{E}[M_{u_1 \ra u_2}(y_1) \mid Y_{u_0} = y_0, Y_{u_1} =y_1] dy_1 = \frac{t^2 f(x_{u_1,u_2})}{4}.
$$
Thus,
$
\mb{E}[M_{u_1 \ra u_2}(y_0) \mid Y_{u_0} = y_0] = \frac{t^2 f(x_{u_1,u_2})}{4} + \int_{t}^{y_0} \mb{E}[M_{u_1 \ra u_2}(y_1) \mid Y_{u_0} = y_0, Y_{u_1} =y_1] dy_1.
$
Now, if we condition on $Y_{u_0} = y_0$, then the executions of \Cref{alg:two_phase} on $G$ and $G \setminus u_0$
proceed identically up until time $y_0$. Thus, $$\mb{E}[M_{u_1 \ra u_2}(y_1) \mid Y_{u_0} = y_0, Y_{u_1} =y_1] = \mb{E}[M^{-u_0}_{u_1 \ra u_2}(y_1) \mid  Y_{u_1} =y_1] \le f(x_{u_1,u_2}) \cdot \UB_{u_1 \ra u_2}^{G \setminus u_0}(y_1, \ell),$$ where the inequality applies \eqref{eqn:induction_two_phase}. It follows that
$$
\mb{E}[M_{u_1 \ra u_2}(y_0) \mid Y_{u_0} = y_0] \le f(x_{u_1,u_2})\left( \frac{t^2}{4} + \int_{t}^{y_0} \UB_{u_1 \ra u_2}^{G \setminus u_0}(y_1, \ell) dy_1 \right).
$$
By applying a symmetric bound to $\mb{E}[M_{u_2 \ra u_1}(y_0) \mid Y_{u_0} = y_0]$, we get
that
$$
\mb{E}[M_{u_1}(y_0) \mid Y_{u_0} = y_0] \le \sum_{u_2 \in N(u_1) \setminus u_0} f(x_{u_1,u_2})\left( \frac{t^2}{2} + \int_{t}^{y_0} \UB_{u_1 \ra u_2}^{G \setminus u_0}(y_1, \ell) dy_1  + \int_{t}^{y_0} \UB_{u_2 \ra u_1}^{G \setminus u_0}(y_2, \ell) dy_2 \right),
$$
and so $\mb{E}[M_{u_0 \ra v_1}(y_0) \mid Y_{u_0} = y_0] /f(x_{u_0,u_1})$ is lower bounded
by 
$$
\LB_{u_0 \ra u_1}^{G}(y_0,\ell):= y_0 - \sum_{u_2 \in N(u_1) \setminus u_0} f(x_{u_1,u_2})\left( \frac{t^2}{2} + \int_{t}^{y_0} \UB_{u_1 \ra u_2}^{G \setminus u_0}(y_1, \ell-1) dy_1  + \int_{t}^{y_0} \UB_{u_2 \ra u_1}^{G \setminus u_0}(y_2, \ell-1) dy_2 \right)
$$
as required. 
\end{proof}

\subsection{Expanding the Recursive Bound} \label{sec:iterate_recursion}
When $\ell$ is even, \Cref{lem:recursion_two_phase} establishes a
lower bound $\LB^{G}_{u_0 \ra  u_1}(y_0, \ell)$
on $\mb{E}[M_{u_0 \ra u_1}(y_0)  \mid Y_{u_0} = y_0]/f(x_0)$ for an
edge $(u_0,v_0) \in E$ of an arbitrary graph $G=(V,E)$ when $y_0 \in (t,1]$. That being said,
it is easy to see that if we make no assumption on $\bm{x}$,
then \Cref{alg:two_phase} attains a selectability of at most $1/2$ for
all $t \in [0,1]$. Thus, we first apply the reduction of
\cite{Fu2021} to ensure that $\bm{x}=(x_e)_{e \in E}$ of $G$ is $1$-regular. In this case, the analysis of \cite{Fu2021} is recovered by
taking $t =0$ and $\ell=4$ in \Cref{lem:recursion_two_phase}, which yields
\begin{equation} \label{eqn:overall_guarantee_recursive}
\frac{\mb{P}[(u_0,u_1) \in \scr{M}]}{x_{u_0,u_1}} \ge a(x_{u_0,u_1})\int_{0}^{1} (\LB_{u_0 \ra u_1}^{G}(y_0, 4) + \LB_{u_1 \ra u_0}^{G}(y_0, 4) ) dy_0
\end{equation}
when applied to \eqref{eqn:overall_guarantee_ezra}.
Using the $1$-regularity of $\bm{x}$ and the analytic properties of $a_0$, \cite{Fu2021} argue that the right-hand side of \eqref{eqn:overall_guarantee_recursive}
is at least $8/15$ after expanding the recursive bound. A natural way to improve on \cite{Fu2021} would be to take $\ell = 6$,
and instead lower bound $\int_{0}^{1} (\LB_{u_0 \ra u_1}^{G}(y_0, 6) + \LB_{u_1 \ra u_0}^{G}(y_0, 6)) dy_0$. Unfortunately, since the analysis only considers vertices at graph distance at most $6$ from $u_0$ or $u_1$, an adversary can construct an input for which $\int_{0}^{1} (\LB_{u_0 \ra u_1}^{G}(y_0, 6) + \LB_{u_1 \ra u_0}^{G}(y_0, 6) )dy_0 = 8/15$ for some edge $(u_0,u_1)$, even when $\bm{x}$ is $1$-regular. While it's possible
that taking $\ell \ge 8$ might improve on $8/15$, the resulting optimization problem becomes increasingly complex, and so we abandon this approach.
Instead, we keep $\ell = 4$, and take $t \in [0,1]$ to get
\begin{equation} \label{eqn:overall_guarantee_recursive_new}
\frac{\mb{P}[(u_0,u_1) \in \scr{M}]}{x_{u_0,u_1}} \ge a(x_{u_0,u_1})\left(\frac{t^2}{2} + \int_{t}^{1} (\LB_{u_0 \ra u_1}^{G}(y_0, 4) + \LB_{u_1 \ra u_0}^{G}(y_0, 4) ) dy_0\right),
\end{equation}
after applying \Cref{lem:recursion_two_phase} to \eqref{eqn:overall_guarantee_ezra}.
Let $t_0 \in (0,1)$ be the unique root of the polynomial
\begin{equation} \label{eqn:t_root}
    4 t^6 + 16 t^5+ 100 t^4 + 180 t^3 + 80 t^2 - 4t-1,
\end{equation}
where $t_0 \approx 0.119$. The following analytic properties of $a_t$ are easily verified:
\begin{proposition}\label{prop:analytic_properties}
For each $z \in [0,1]$, we have
$f_t(z) = \frac{z(3 + 2 t^5 - 5 t^2)}{3 + 2 t^5 (1 - z) + 2 z + 10 t^3 z - 5 t^2 (1 + 2 z)}$.
\end{proposition}
\begin{proposition} \label{prop:analytic_two_values}
Suppose $t \le t_0$, where $t_0$ satisfies \eqref{eqn:t_root}. 
Then for each $x, y \in [0,1]$,
\begin{align*}
f_t(x) \left(\frac{1}{3} + \frac{t^3}{6} - \frac{t^2}{2}
-\frac{(2- 2x - y)(3 + 2 t^5- 5 t^2)}{60}\right) & + f_t(y) \left(\frac{1}{3} + \frac{t^3}{6} - \frac{t^2}{2}
-\frac{(2- 2y - x)(3 + 2 t^5- 5 t^2)}{60}\right) \\
&\le \left(\frac{1}{3} + \frac{t^3}{6} - \frac{t^2}{2} -  \frac{(3 + 2 t^5- 5 t^2)}{30}\right)(x + y).
\end{align*}
\end{proposition}
For each fixed $t \in [0,t_0]$, we now use Propositions \ref{prop:analytic_properties} and \ref{prop:analytic_two_values},
to lower bound the right-hand side of \eqref{eqn:overall_guarantee_recursive_new}. As in the previously
described case of $t =0$ and $\ell =6$, our analysis of the right-hand side of \eqref{eqn:overall_guarantee_recursive_new} is tight.
\begin{lemma} \label{lem:expand_recursion}
If $\bm{x}=(x_e)_{e \in E}$ of $G =(V,E)$ is $1$-regular and $(u_0,u_1) \in E$, then for each $t \in [0,t_0]$,
$$
a(x_{u_0,u_1})\left(\frac{t^2}{2} + \int_{t}^{1} (\LB_{u_0 \ra u_1}^{G}(y_0, 4) + \LB_{u_1 \ra u_0}^{G}(y_0, 4) ) dy_0\right) \ge \frac{1}{30} (16 + 5 t^2 - 10 t^3 + 4 t^5).
$$
\end{lemma}
Assuming \Cref{lem:expand_recursion}, \Cref{thm:general_graph_positive} follows by setting $t = t_0$ in \Cref{lem:expand_recursion} and using \eqref{eqn:overall_guarantee_recursive_new} to attain a selectability
bound of $\frac{1}{30} (16 + 5 t_0^2 - 10 t_0^3 + 4 t_0^5) \approx 0.535$.
\begin{proof}[Proof of \Cref{lem:expand_recursion}]
We begin by lower bounding $\int_{t}^{1} \LB_{u_0 \ra u_1}^{G}(y_0,4) dy_0$ by
\begin{equation} \label{eqn:u_1}
\frac{1}{2} - \frac{t^2}{2} - \frac{x_{u_0,u_1}(3 + 2 t^5 - 5 t^2)}{60} - \sum_{u_2 \in N(u_1) \setminus u_0} 
 f(x_{u_1,u_2}) \left(\frac{1}{3} + \frac{t^3}{6} - \frac{t^2}{2}
-\frac{(2- 2x_{u_2,u_0} - x_{u_2,u_1})(3 + 2 t^5- 5 t^2)}{60}\right)
\end{equation}
and $\int_{t}^{1} \LB_{u_1 \ra u_0}^{G}(y_0,4) dy_0$ by
\begin{equation} \label{eqn:u_0}
\frac{1}{2} - \frac{t^2}{2} - \frac{x_{u_0,u_1}(3 + 2 t^5 - 5 t^2)}{60} - \sum_{u_2 \in N(u_0) \setminus u_1} 
 f(x_{u_1,u_2}) \left(\frac{1}{3} + \frac{t^3}{6} - \frac{t^2}{2}
-\frac{(2- 2x_{u_2,u_1} - x_{u_2,u_0})(3 + 2 t^5- 5 t^2)}{60}\right).
\end{equation}
Before proving these lower bounds, we first
show how they allow us to complete the proof of the lemma. Observe
that for $t \in [0,t_0]$, \Cref{prop:analytic_two_values} implies that
\begin{align*}
&\sum_{u_2 \in N(u_1) \setminus u_0} f(x_{u_1,u_2}) \left(\frac{1}{3} + \frac{t^3}{6} - \frac{t^2}{2}
-\frac{(2- 2x_{u_2,u_0} - x_{u_2,u_1})(3 + 2 t^5- 5 t^2)}{60}\right)  \\
&+ \sum_{u_2 \in N(u_0) \setminus u_1} f(x_{u_1,u_2}) \left(\frac{1}{3} + \frac{t^3}{6} - \frac{t^2}{2}
-\frac{(2- 2x_{u_2,u_1} - x_{u_2,u_0})(3 + 2 t^5- 5 t^2)}{60}\right) \\
&\le \left(\frac{1}{3} + \frac{t^3}{6} - \frac{t^2}{2}-  \frac{(3 + 2 t^5- 5 t^2)}{30}\right)(2 - 2x_{u_0,u_1}),
\end{align*}
where we have used that $\sum_{u_2 \neq u_0, u_1} ( x_{u_1,u_2} + x_{u_0,u_2}) = 2 -2x_{u_1,u_0}$.
Thus, $\frac{t^2}{2} + \int_{t}^{1} (\LB_{u_0 \ra u_1}^{G}(y_0, 4) + \LB_{u_1 \ra u_0}^{G}(y_0, 4)) dy_0$ is lower bounded by
$$
1 - \frac{t^2}{2}-  2\left(\frac{1}{3} + \frac{t^3}{6} - \frac{t^2}{2} -  \frac{(3 + 2 t^5- 5 t^2)}{30}\right) + x_{u_0,u_1}\left(\frac{11}{30} - \frac{t^2}{2} + \frac{t^3}{3} - \frac{t^5}{5}\right),
$$
after applying the simplification 
$$2\left(\frac{1}{3} + \frac{t^3}{6} - \frac{t^2}{2} -  \frac{(3 + 2 t^5- 5 t^2)}{30}\right)- \frac{(3 + 2 t^5 - 5 t^2)}{60} = \frac{11}{30} - \frac{t^2}{2} + \frac{t^3}{3} - \frac{t^5}{5}.$$
The function $x_{u_0,u_1} \ra a(x_{u_0,u_1})\left(1 - \frac{t^2}{2}-  2\left(\frac{1}{3} + \frac{t^3}{6} - \frac{t^2}{2} -  \frac{(3 + 2 t^5- 5 t^2)}{30}\right) + x_{u_0,u_1}\left(\frac{11}{30} - \frac{t^2}{2} + \frac{t^3}{3} - \frac{t^5}{5}\right)\right)$ is minimized at $x_{u_0,u_1} = 0$ for each fixed $t \in [0,1]$. Since $a(0)=1$ and $1 - \frac{t^2}{2}-  2\left(\frac{1}{3} + \frac{t^3}{6} - \frac{t^2}{2} -  \frac{(3 + 2 t^5- 5 t^2)}{30}\right)= \frac{1}{30} (16 + 5 t^2 - 10 t^3 + 4 t^5)$, this implies
the bound claimed by the lemma.

We now prove \eqref{eqn:u_1} and \eqref{eqn:u_0}. Since the terms are symmetric, it suffices to
only prove \eqref{eqn:u_1}. First observe that by iterating
the recursive definition of
$\LB_{u_0 \ra u_1}^{G}(y_0,4)$ and integrating over $y_0 \in [0,1]$,
$$
\int_{t}^{1} \LB_{u_0 \ra u_1}^{G}(y_0,4) dy_0 = \frac{1}{2} - \frac{t^2}{2}  - \sum_{u_2 \neq u_0} 
 f(x_{u_1,u_2}) \left(\frac{(1-t)t^2}{2} +  \int_{t}^{1} \int_{t}^{y_0} (\UB^{G \setminus u_0}_{u_1 \ra u_2}(y_1, 3) + \UB^{G \setminus u_0}_{u_2 \ra u_1}(y_1, 3) ) dy_1 dy_0\right).
$$
We focus on upper bounding $\int_{t}^{1} \int_{t}^{y_0} \UB^{G \setminus u_0}_{u_1 \ra u_2}(y_1, 3)dy_1 dy_0$, as the upper bound on $\int_{t}^{1} \int_{t}^{y_0} \UB^{G \setminus u_0}_{u_2 \ra u_1}(y_1, 3) dy_1 dy_0$
will follow by reversing the roles of $u_1$ and $u_2$. Now, by expanding the recursion twice more,
$$
\UB^{G \setminus u_0}_{u_1 \ra u_2}(y_1, 3) = y_1  - \sum_{u_3  \in N(u_2) \setminus u_0,u_1} f(x_{u_2,u_3})  \left(\frac{t^2}{2} +  \int_{t}^{y_1} (\LB^{G \setminus u_0,u_1}_{u_2 \ra u_3}(y_2, 2)  + \LB^{G \setminus u_0,u_1}_{u_3 \ra u_2}(y_2, 2) ) dy_2\right),
$$
and
$$
\LB^{G \setminus u_0,u_1}_{u_2 \ra u_3}(y_2, 2) = y_2 - \sum_{u_4  \in N(u_3) \setminus u_0,u_1,u_2} f(x_{u_3,u_4})  \left(\frac{t^2}{2} +  \int_{t}^{y_2} (\UB^{G \setminus u_0,u_1,u_2}_{u_3 \ra u_4}(y_3, 1)  + \UB^{G \setminus u_0,u_1,u_2}_{u_4 \ra u_3}(y_3, 1) ) dy_3\right).
$$
But, $\UB^{G \setminus u_0,u_1,u_2}_{u_3 \ra u_4}(y_3, 1)  + \UB^{G \setminus u_0,u_1,u_2}_{u_4 \ra u_3}(y_3, 1) = 2y_3$,
so
\begin{align*}
    \LB^{G \setminus u_0,u_1}_{u_2 \ra u_3}(y_2, 2) &= y_2 - \sum_{u_4  \in N(u_3) \setminus u_0,u_1,u_2}f(x_{u_3,u_4})\left(y^2_2 -\frac{t^2}{2} \right) \\
    &\ge y_2 - ( 1- x_{u_2,u_3})\left(y^2_2 -\frac{t^2}{2} \right),
\end{align*}
where the inequality uses $f(z) \le z$, and the $1$-regularity of $\bm{x}$. Similarly,
$$
\LB^{G \setminus u_0,u_1}_{u_3 \ra u_2}(y_2, 2) \ge y_2 - ( 1- x_{u_2,u_3})\left(y^2_2 -\frac{t^2}{2} \right).
$$
By applying these bounds to $\UB^{G \setminus u_0}_{u_1 \ra u_2}(y_1, 3)$, after simplifcation
we get that
\begin{align*}
\UB^{G \setminus u_0}_{u_1 \ra u_2}(y_1, 3) &\le y_1 - \sum_{u_3  \in N(u_2) \setminus u_0,u_1} f(x_{u_2,u_3})  \left(\frac{t^2}{2} +  \int_{t}^{y_1} 2 \left(y_2 - ( 1- x_{u_2,u_3})\left(y^2_2 -\frac{t^2}{2} \right) \right) dy_2 \right)\\
&= y_1 - \sum_{u_3  \in N(u_2) \setminus u_0,u_1} f(x_{u_2,u_3}) \left(y_1^2 - \frac{t^2}{2}- \frac{(1-x_{u_2,u_3})}{3}(t^3 - 3 t^2 y_1 + 2 y_1^3) \right).
\end{align*}
We are now ready to upper bound $\int_{t}^{1} \int_{t}^{y_0} \UB^{G \setminus u_0}_{u_1 \ra u_2}(y_1, 3) dy_1 dy_0$. First observe that
$$\int_{t}^{1} \int_{t}^{y_0} y_1 dy_1 dy_0 =  \frac{1}{2} \left( \frac{1}{3} + \frac{2t^3}{3} - t^2\right)$$
and $\int_{t}^{1} \int_{t}^{y_0}\left(y_1^2 - \frac{t^2}{2}- \frac{(1-x_{u_2,u_3})}{3}(t^3 - 3 t^2 y_1 + 2 y_1^3) \right) dy_1 dy_0$ is equal to 
$$
\frac{1}{60}(3 + 2 t^5 (1 - x_{u_2,u_3}) + 2 x_{u_2,u_3} + 10 t^3 x_{u_2,u_3} - 5 t^2 (1 + 2 x_{u_2,u_3})).
$$
Thus, $\int_{t}^{1} \int_{t}^{y_0}\UB^{G \setminus u_0}_{u_1 \ra u_2}(y_1, 3) dy_0 dy_1$ is upper bounded
by
$$
\frac{1}{2} \left( \frac{1}{3} + \frac{2t^3}{3} - t^2\right)  - \sum_{u_3 \neq u_0,u_1} f(x_{u_2,u_3}) \frac{1}{60}(3 + 2 t^5 (1 - x_{u_2,u_3}) + 2 x_{u_2,u_3} + 10 t^3 x_{u_2,u_3} - 5 t^2 (1 + 2 x_{u_2,u_3})).
$$
Now, $f(z) = \frac{z(3 + 2 t^5 - 5 t^2)}{3 + 2 t^5 (1 - z) + 2 z + 10 t^3 z - 5 t^2 (1 + 2 z)}$ by \Cref{prop:analytic_properties}, and so
\begin{equation}\label{eqn:u1_to_u2}
\int_{t}^{1} \int_{t}^{y_0}\UB^{G \setminus u_0}_{u_1 \ra u_2}(y_1, 3) dy_1 dy_0 \le
\frac{1}{2} \left( \frac{1}{3} + \frac{2t^3}{3} - t^2\right)- \frac{(1- x_{u_2,u_0} - x_{u_2,u_1})(3 + 2 t^5 - 5 t^2)}{60},
\end{equation}
where we've used the $1$-regularity of $\bm{x}$. 
By reversing the roles of $u_1$ and $u_2$, we also get that 
\begin{equation}\label{eqn:u2_to_u1}
\int_{t}^{1} \int_{t}^{y_0}\UB^{G \setminus u_0}_{u_2 \ra u_1}(y_1, 3) dy_1 dy_0 \le \frac{1}{2} \left( \frac{1}{3} + \frac{2t^3}{3} - t^2\right) - \frac{(1- x_{u_1,u_0} - x_{u_2,u_0})(3 + 2 t^5 - 5 t^2)}{60}.
\end{equation}
Returning to $\int_{t}^{1} \LB_{u_0 \ra u_1}^{G}(y_0,4) dy_0$, since $\frac{1}{3} + \frac{2t^3}{3} - t^2 + \frac{(1-t)t^2}{2} = \frac{1}{3} + \frac{t^3}{6} - \frac{t^2}{2}$, we can apply \eqref{eqn:u1_to_u2} and \eqref{eqn:u2_to_u1}
to lower bound $\int_{t}^{1} \LB_{u_0 \ra u_1}^{G}(y_0,4) dy_0$ by
$$
\frac{1}{2} - \frac{t^2}{2} - \sum_{u_2  \in N(u_1) \setminus u_0} 
 f(x_{u_1,u_2}) \left(\frac{1}{3} + \frac{t^3}{6} - \frac{t^2}{2}
-\frac{(2- 2x_{u_2,u_0} - x_{u_1,u_0} - x_{u_2,u_1})(3 + 2 t^5 - 5 t^2)}{60}\right).
$$
Now, $\sum_{u_2  \in N(u_1) \setminus u_0} f(x_{u_1,u_2}) \le 1$, so we can remove the $x_{u_0,u_1}$ term from the summand and lower bound $\int_{t}^{1} \LB_{u_0 \ra u_1}^{G}(y_0,4) dy_0$ by
$$
\frac{1}{2} - \frac{t^2}{2} - \frac{x_{u_0,u_1}(3 + 2 t^5 - 5 t^2)}{60} - \sum_{u_2 \in N(u_1) \setminus u_0} 
 f(x_{u_1,u_2}) \left(\frac{1}{3} + \frac{t^3}{6} - \frac{t^2}{2}
-\frac{(2- 2x_{u_2,u_0} - x_{u_2,u_1})(3 + 2 t^5- 5 t^2)}{60}\right),
$$
which proves \eqref{eqn:u_1}, and thus the lemma.
\end{proof}
We complete the section by discussing the case of vanishing edges values,
as mentioned in \Cref{rem:general_graph_vanishing}. Crucially, the reduction of \cite{Fu2021} to make
$\bm{x}$ $1$-regular preserves vanishing edges values. Thus, we may assume that $x_e \le \eps$ for all $e \in E$, where $\eps(n)=\eps$ 
satisfies $\eps(n) \rightarrow 0$ as $n = |V| \rightarrow \infty$. 
The inequality of \Cref{prop:analytic_two_values} is then true up to a multiplicative error which tends to $0$
as $n \rightarrow \infty$
for \textit{all} $t \in [0,1]$. The same computations in the proof of \Cref{lem:expand_recursion} then apply, and we get that
$$
\frac{\mb{P}[e \in \scr{M}]}{x_{u,v}} \ge (1 - o_n(1)) \frac{1}{30} (16 + 5 t^2 - 10 t^3 + 4 t^5),
$$
where $o_{n}(1)$ is a function which tends to $0$ as $n \rightarrow \infty$.
The function $t \rightarrow \frac{1}{30} (16 + 5 t^2 - 10 t^3 + 4 t^5)$
is maximized on $[0,1]$ at $t = \frac{1}{2}(\sqrt{3}-1) \approx 0.366$ when it takes the value 
$\frac{4}{5} - \frac{3 \sqrt{3}}{20} \approx 0.540$.

\section{Proving \Cref{thm:hardness_vertex_rom}}

In order to prove \Cref{thm:hardness_vertex_rom},
we analyze the complete bipartite graph with
$2 n$ vertices and uniform edge values.
Let $G=(L,R,E)$ where $E = L \times R$, and $|L|= |R| =n$ for $n \ge 1$,
and set $x_e = 1/n$ for all $e \in E$. We work in the asymptotic setting as $n \rightarrow \infty$,
and say that a sequence of events $(\scr{E}_n)_{n \ge 1}$ occurs \textit{with high probability} (w.h.p.), provided $\mb{P}[\scr{E}_n] \rightarrow 1$ as $n \rightarrow \infty$.

\begin{lemma}\label{lem:expectation_negative}
For any RCRS which outputs matching $M$ on $G$, $\E[|M|] \le \frac{(1 + o(1))(1 + e^2)n}{2e^2}$.
\end{lemma}
Assuming \Cref{lem:expectation_negative}, \Cref{thm:hardness_vertex_rom} then follows immediately.
\begin{proof}[Proof of \Cref{thm:hardness_vertex_rom} using \Cref{lem:expectation_negative}]
Suppose that we fix an arbitrary RCRS which is $c$-selectable for $c \ge 0$,
and let $M$ be the matching it creates when executing on $G$. 
Clearly,
$\mb{E}[ |M|]/n \ge c$ by definition of $c$-selectability. 
By applying \Cref{lem:expectation_negative},
and taking $n\to\infty$, we get that $c\le (1+\sfe^{-2})/2$.
\end{proof}

To prove \Cref{lem:expectation_negative}, we consider an algorithm for maximizing $\bE[|M|]$ and show that the number of edges in the matching cannot exceed $\frac{(1 + o(1))(1 + \sfe^{-2})n}{2}$ in expectation.  Without loss of generality, we can assume such an algorithm is deterministic.

For each \textit{round} $1 \le i \le 2 n$, let $v_i$ be the $i^{th}$ vertex of $G$ presented to the RCRS.
Observe that if $V_t := \{v_1, \ldots , v_t\}$, then conditional on $V_t$, $v_{t+1}$ is distributed
uniformly at random (u.a.r.) amongst $V(G) \setminus V_t$ for $0 \le t \le 2n -1$. Let $F_{t+1} := F_{v_{t+1}}$ be the vertex chosen by $v_{t+1}$.
Observe that $\mb{P}[ F_{t+1} = v_i \mid V_t, v_{t+1}] = 1/n$ for each $v_i \in N(v_{t+1})$.

If $\scr{M}_t$ is the matching
constructed by the RCRS after $t$ rounds, then since the RCRS is deterministic,
$\scr{M}_t$ is a function of $(v_i, F_i)_{i=1}^{t}$. Thus, $\scr{M}_t$
is measurable with respect to $\scr{H}_t$, the sigma-algebra generated from $(v_i, F_i)_{i=1}^{t}$
(here $\scr{H}_{0}:= \{ \emptyset, \Omega\}$, the trivial sigma-algebra).
We refer to $\scr{H}_t$ as the \textit{history} after $t$ steps.

Let $m(s):= (\sfe^{-s} + s -1)/2$ for each real $s \ge 0$. Note that $m$ is the unique solution
to the differential equation $m'(s) = s/2- m(s)$
with initial condition $m(0)=0$.
Roughly speaking,
we will prove that w.h.p., the random variable $M(t)/n$ is upper bounded by  $(1 + o(1)) m(t/n)$ for each $0 \le t \le 2n$,
where $M(t)$ is the number of edges in $\scr{M}_t$.

\begin{proposition} \label{prop:de_dominance}
Fix any constant $0 \le \eps < 2$. Then, $\frac{M(t)}{n} \le (1 + o(1)) m(t/n)$ for all $0 \le t \le \eps n$ with probability at least $1 - o(1/n)$.
\end{proposition}

We emphasize that in \Cref{prop:de_dominance}, a constant $\eps$ is fixed first, and $n$ is taken to $\infty$ afterward.
As we can take constant $\eps$ to be arbitrarily close to $2$, the $(1+\sfe^{-2})/2$ upper bound in \Cref{lem:expectation_negative} follows easily:

\begin{proof}[Proof of \Cref{lem:expectation_negative} using \Cref{prop:de_dominance}]
Fix $0 \le \eps < 2$. Observe that \Cref{prop:de_dominance} implies
\begin{equation}
    \frac{\E[ M(\eps n)]}{n} \le  (1 - o(1/n))(1 + o(1)) m(\eps) + o(1/n)\eps = (1 + o(1)) m(\eps),
\end{equation}
where the $o(1/n)\eps$ term uses the trivial upper bound that $M(\eps n) \le \eps n$.
Moreover, the same bound yields $\E[ M(2 n) - M(\eps n)] \le (2 - \eps) n $. Thus, since $m(\eps)$
is increasing in $\eps$,
\[
    \frac{\E[M(2n)]}{n} \le (1 + o(1))\frac{(\sfe^{-\eps} + \eps -1)}{2} + (2 - \eps) \le \frac{(1 + \sfe^{-2})}{2} + (2 - \eps),
\]
for all $n$ sufficiently large. Since this holds for each $0 \le \eps < 2$,  we get that
$\frac{\E[ M(2n)]}{n} \le (1 + o(1)) \frac{(1 + \sfe^{-2})}{2}$.
As $\E[ M(2n)]$ is an upper bound on the expected
size of any matching created by an RCRS, the proof is complete.
\end{proof}

In order to prove \Cref{prop:de_dominance}, for each constant $0 \le \eps < 2$,
and $0 \le t \le \eps n$, 
we first upper bound the expected one-step changes of $M(t)$, conditional
on the current history $\scr{H}_t$. More formally, we upper bound
$\E[ \Delta M(t) \mid \scr{H}_t]$, where $\Delta M(t):= M(t+1) - M(t)$.
Our goal is to show that \begin{equation} \label{eqn:de_expected_change}
\E[ \Delta M(t) \mid \scr{H}_t] \le \frac{t}{2n} -\frac{M(t)}{n}.
\end{equation}
This bound only holds for certain instantiations of the random variables
$(v_i,F_i)_{i=1}^t$. Specifically, there must be approximately the same
number of vertex arrivals in $R$ as in $L$ by step $t$.
We formalize this using the event $Q_t$,
which is defined to occur when $|R \setminus R_t| \le (1 + n^{-1/3}) (2n -t)/2$
and $|L \setminus L_t| \le (1 + n^{-1/3}) (2n -t)/2$.
We first argue that $Q_t$ occurs w.h.p. for all $0 \le t \le \eps n$.
\begin{lemma} \label{lem:high_probability_event}
$\mb{P}[\cap_{t=0}^{\eps n} Q_t] \ge 1 - o(1/n)$.
    
\end{lemma}

\begin{proof}
Fix $0 \le t \le \eps n$.
Observe that since $V_t$ is a uniform at random subset of $V$ of size $t$,
$R_t := |R \cap V_t|$ is
distributed as a hyper-geometric random variable on a universe of size
$2 n$ with $t$ trials and success probability $1/2$.
We denote this by $|R_t| \sim \text{Hyper}(2n ,n, t)$. Similarly,
$L_t:=|L \cap V_t| \sim \text{Hyper}(2n ,n, t)$.

Now, $\text{Hyper}(N,K, j)$ is at least as concentrated as $\Bin(j, K/N)$ (see Chapter~21 in~\cite{Frieze2015} for details).
Since $\eps <2$, we can apply standard Chernoff bounds to $R_t$ and $L_t$ to get that $\mb{P}[Q_t] \ge 1 - o(1/n^2)$. The result follows by applying a union bound over $0 \le t \le \eps n$.
\end{proof}

\begin{lemma} \label{lem:matching_expected_difference}
For each $0 \le t \le  \eps n$, if $Q_t$ occurs, then
$\E[ \Delta M(t) \mid \scr{H}_t] \le (1 + n^{-1/3}) \left( \frac{t}{2n} - \frac{M(t)}{n} \right)$.
\end{lemma}

\begin{proof}
First observe that
$\E[ \Delta M(t) \mid \scr{H}_t]  = \mb{P}[ v_{t+1} \in \scr{M}_{t+1} \mid \scr{H}_t].$
Now,  $\mb{P}[ v_{t+1} \in \scr{M}_{t+1} \mid \scr{H}_t]$ is equal to
$$
\mb{P}[ v_{t+1} \in \scr{M}_{t+1} \mid \scr{H}_t, v_{t+1} \in R] \cdot \mb{P}[ v_{t+1} \in R \mid \scr{H}_t] + \mb{P}[ v_{t+1} \in \scr{M}_{t+1} \mid \scr{H}_t, v_{t+1} \in L ] \cdot \mb{P}[ v_{t+1} \in L \mid \scr{H}_t]
$$
Clearly, $\mb{P}[ v_{t+1} \in R \mid \scr{H}_t] = |R \setminus R_t|/(2n -t)$
and $\mb{P}[ v_{t+1} \in L \mid \scr{H}_t] = |L \setminus L_t|/(2n -t)$. On the other hand,
if $Q_t$ occurs, then $\max\{ |R \setminus R_t|, |L \setminus L_t|\} \le (1 + n^{-1/3}) (2n -t)/2$.
Thus, if $Q_t$ occurs, then we can upper bound $\mb{P}[ v_{t+1} \in \scr{M}_{t+1} \mid \scr{H}_t]$
by
\[
  \frac{(1 + n^{-1/3})}{2} \left( \mb{P}[ v_{t+1} \in \scr{M}_{t+1} \mid \scr{H}_t, v_{t+1} \in R] + \mb{P}[ v_{t+1} \in \scr{M}_{t+1} \mid \scr{H}_t, v_{t+1} \in L ] \right)
\]
It remains to bound $\mb{P}[ v_{t+1} \in \scr{M}_{t+1} \mid \scr{H}_t, v_{t+1} \in R]$
and $\mb{P}[ v_{t+1} \in \scr{M}_{t+1} \mid \scr{H}_t, v_{t+1} \in L ]$. We begin with the first
expression.

Observe that if we condition on $\scr{H}_t$ and $v_{t+1} \in R$, then 
$v_{t+1}$ can be added to the matching if and only if $F_{t+1} =u_{t+1}$ 
is not matched by $\scr{M}_t$. On the other hand, conditional
on $v_{t+1} \in R$, $u_{t+1}$ is distributed u.a.r. in $L$. Thus,
$$\mb{P}[ v_{t+1} \in \scr{M}_{t+1} \mid \scr{H}_t,  v_{t+1} \in R] \le \frac{|L_t| - |\scr{M}_t|}{n}.$$
and by symmetry,
$
\mb{P}[ v_{t+1} \in \scr{M}_{t+1} \mid \scr{H}_t,  v_{t+1} \in L] \le \frac{|R_t| - |\scr{M}_t|}{n}.
$
Combining everything together,
\begin{align*}
\mb{P}[ v_{t+1} \in \scr{M}_{t+1} \mid \scr{H}_t] \cdot \bm{1}_{Q_t} &\le \frac{(1 + n^{-1/3})}{2} \left(\frac{|L_t| - |\scr{M}_t|}{n} + \frac{|R_t| - |\scr{M}_t|}{n}\right) \\
&= (1 + n^{-1/3})\left( \frac{t}{2n} - \frac{M(t)}{n} \right),
\end{align*}
where the final equation uses $|L_t| + |R_t| =t$. 
\end{proof}
In order to complete the proof of \Cref{prop:de_dominance}, consider $(M(t)/n)_{t=0}^{\eps n}$. Observe that if we scale $t$ by $n$,
then we can view $\mb{E}[ \Delta M(t)]$ as the ``derivative'' of
$(M(t)/n)_{t=0}^{\eps n}$ evaluated at $t/n$. Thus, \Cref{lem:high_probability_event,lem:matching_expected_difference}  imply that 
$(M(t)/n)_{t=0}^{\eps n}$ satisfies the \textit{differential inequality}, $r'(s) \le s/2 - r(s)$ with $r(0)=0$. Intuitively, since $m(s)$ satisfies the corresponding differential \textit{equation}, $m' = s/2 - m(s)$ with $m(0)=0$,
this suggests that $M(t)/n$ ought to be dominated by $m(t/n)$. This is
precisely the statement of \Cref{prop:de_dominance}, and assuming the above lemmas,
follows from the general purpose ``one-sided'' differential equation method of \citet{bennett2023}.
We provide the details
in \Cref{sec:de_dominance}.

\appendix 

\section{Deferred Proofs}\label{appendix}

\subsection{Rank-$1$ RCRS} \label{sec:rank_one}

Consider set $V$ with fractional values $(x_v)_{v \in V}$ satisfying $\sum_{v \in V} x_v =1$. Each $v$ has an independently drawn indicator random variable $X_v \sim \Ber(x_v)$ for the event $v$ is active, and an arrival time $Y_v$ drawn uniformly from $[0,1]$. An RCRS can select at most one active $v \in V$. For an arbitrary RCRS, we define $S(y)$ to be the indicator random variable for the event that no vertex is selected by time $y \in [0,1]$.

Recall the recursively defined RCRS from  \Cref{sec:introContInd} using $c(y) := \sfe^{-y}$: Upon the arrival
of $v$ with $Y_v = y$, the RCRS draws an independent random bit $A_v$ which
is $1$ with probability
$
c(y)/\mb{E}[S(y) \mid Y_v = y].
$
Then, if no selection was made before time $Y_v$, and $A_v \cdot X_v = 1$,
the RCRS selects $v$. This RCRS has the property that $\mb{P}[\text{$v$ is selected} \mid Y_v = y] = c(y)$ for each $v \in V$ and $y \in [0,1]$.
 We claim that for all $y \in [0,1]$,
$$
    \mb{E}[ S(y) \mid Y_v = y] = \sfe^{-y (1 - x_{v})}.
$$
We provide a proof sketch that verifies this via continuous induction. A formal
proof requires defining and analyzing the algorithm when
there is a discretization parameter $T$ (as done in \Cref{alg:recursive_rcrs}).

\begin{proof}

Fix $t \in [0,1]$, and assume that for each $y < t$ and $v \in V$, it holds
that 
$$
    \mb{E}[ S(y) \mid Y_v = y] = \sfe^{-y (1- x_{v,u}) }.
$$
We shall verify that this holds at time $t$ for each $v$.
Let us say that $v'$ survives, provided $A_{v'} X_{v'} =1$.
Observe that conditional on $Y_v = t$, $S(t) =0$ if and only if
each $v' \neq v$ either does \textit{not} survive, or has $Y_{v'} > t$.
Thus, using the independence over $v' \neq v$, we get that
$$
    \mb{E}[S(t) \mid Y_v =t] = \prod_{v' \neq v} \left(1 - \mb{P}[ A_{v'} X_{v'} =1, Y_{v'} < t \mid Y_v =t] \right).
$$
On the other hand, using the independence of events $Y_v = t$
and $\{A_{v'} X_{v'} =1\} \cap \{Y_{v'} < t\}$,
$$
\mb{P}[A_{v'} X_{v'} =1, Y_{v'} < t \mid Y_v = t] = \mb{P}[A_{v'} X_{v'} =1, Y_{v'} < t] =  x_{v'} \int_{0}^{t} \frac{c(y) dy}{\mb{P}[S(y) \mid Y_{v'} = y]}.
$$
Now, after applying the induction hypothesis and integrating,
$
   \int_{0}^{t} \frac{c(y) dy}{\mb{P}[S(z) \mid Y_{v'} = y]} = 1 - \sfe^{- t x_{v'}}.
$
After combining all equations, we get the desired guarantee,
$
\mb{E}[S(t) \mid Y_v =t] = \prod_{v' \neq v} \sfe^{- t x_{v'}} = \sfe^{-t (1 -x_v)},
$
thereby verifying that the induction hypothesis holds at time $t$.
\end{proof}

\subsection{Closed-form Expression for $\alpha_g$}
Recall that the \textit{upper incomplete gamma function} is defined as $\Gamma(s,z):= \int_{z}^{\infty} \zeta^{s-1} \sfe^{-\zeta} d\zeta$ for $s >0$ and $z \in \mb{R}$, where $\Gamma(s,0) = \int_{0}^{\infty} \zeta^{s-1} \sfe^{-\zeta} d\zeta$ is the usual definition of the \textit{gamma function}. Now,
suppose that \Cref{alg:recursive_rcrs} is given a graph $G$ with odd-girth $g \ge 3$. In this
case, we use the selection function $c: [0,1] \ra [0,1]$, where and $c(0):=1$ and for each $y \in (0,1]$,
$$c(y) := \frac{1}{2y} - \frac{\sfe^{-2 y}}{2y} - \frac{\sfe^{-2 y} ( \Gamma(g, -2y) - \Gamma(g, 0))}{2^{g-1} y(g -1)!}.$$
Note that since $g$ is odd,  $\Gamma(g, -2y) - \Gamma(g, 0)) \ge 0$.
Now, $\alpha_g := \int_{0}^{1} c(y) 2y dy$, and $\int_{0}^{1} (1 - \sfe^{-2y}) dy = \frac{1}{2} + \frac{1}{2 \sfe^2}$, thus
$$\alpha_g = \frac{1}{2} + \frac{1}{2 \sfe^2} - \frac{1}{2^{g-2} (g -1)!}\int_{0}^{1} \sfe^{-2 y} ( \Gamma(g, -2y) - \Gamma(g, 0)) dy.$$
Moreover, if we change the order of integration when evaluating $\int_{0}^{1} \int_{-2y}^{\infty} \sfe^{-2y} \zeta^{s-1}  \sfe^{-\zeta} d\zeta dy$ and apply integration by parts, then since $g$ is odd,
$$
\frac{1}{2^{g-2} (g -1)!} \int_{0}^{1} \sfe^{-2 y} ( \Gamma(g, -2y) - \Gamma(g, 0)) dy = 
\frac{1}{(g-1)!} \left( \frac{2}{g} - \frac{\Gamma(g, -2) - \Gamma(g,0)}{ 2^{g-1} \sfe^2} \right)
$$
and so we get that 
\begin{equation} \label{eqn:gamma_explicit}
    \alpha_g = \frac{1}{2} + \frac{1}{2 \sfe^2} - \frac{1}{(g-1)!} \left( \frac{2}{g} - \frac{\Gamma(g, -2) - \Gamma(g,0)}{ 2^{g-1} \sfe^2} \right).
\end{equation}
Clearly, $\alpha_{g} \ra \frac{1}{2} + \frac{1}{2 \sfe^2}$ as $g \ra \infty$.
It is not hard to show that if $s_g:= \frac{2}{g} - \frac{\Gamma(g, -2) - \Gamma(g,0)}{ 2^{g-1} \sfe^2}$,
then $(s_{2k+1})_{k =1}^{\infty}$ is a decreasing sequence of positive terms. Thus, $(\alpha_{2k+1})_{k=1}^{\infty}$ is an increasing sequence as claimed in \Cref{prop:increasing_alpha}.

\subsection{Proof of \Cref{prop:de_dominance} using \Cref{lem:matching_expected_difference}} \label{sec:de_dominance}

We apply Theorem $3$ from \citet{bennett2023}, where
we have intentionally chosen the below notation
to match the parameters of [Theorem $3$,\citet{bennett2023}]. Define $\scr{D}= [0,2] \times [0,1]$, and set $f(s, r) := s/2 - r$ for $(s,r) \in \scr{D}$. Recall that $m(s)=(\sfe^{-s} + s -1)/2$ is the unique solution to the differential equation $m'(s) = f(s,m(s))$ with initial condition $m(0)=0$. We apply
the theorem to a single sequence of random variables $(M(i)/n)_{i=0}^{\eps n}$ (i.e., $a=1$), the filtration $(\scr{H}_i)_{i=0}^{\eps n}$, and the above system. 

We first set the analytic parameters of the system. We can take $L :=2$, and $B:=1$, since $f$ has Lipschitz constant at most $2$, and $|f| \le 1$. More, we can take $\sigma := \eps$, as $m$ is a solution to
the system on the interval $[0, \eps]$ (in fact, it satisfies the differential equation on the entire real-line).

We now consider $(M(i)/n)_{i=0}^{2n}$. First recall the sequence of events $(Q_i)_{i=0}^{2 n}$ defined before \Cref{lem:matching_expected_difference}. Let $I$ be defined as the first $t \ge 0$ such that the event $Q_t$ does \textit{not} hold. Clearly, $I$ is a stopping time with respect to $(\scr{H}_i)_{i=0}^{2 n}$. It remains to verify `Initial Condition', `Boundedness Hypothesis' and `Trend Hypothesis' from Theorem $3$.
The `Initial Condition' clearly holds, as $M(0) = m(0) =0$. The `Boundedness Hypothesis' holds for $\beta = b =1$ as $M(i)$ changes by at most one in each step. Finally, the `Trend Hypothesis' holds with $\delta = O(n^{-1/3})$ for all $0 \le i \le \min\{I, \sigma n\}$, due to \Cref{lem:matching_expected_difference}.  
By taking $\lambda = n^{2/3}$, we get $\lambda \ge \max\{ B + \beta, \frac{L + BL + \delta n}{3 L} \}$ for $n$ sufficiently large, and so [Theorem $3$, \cite{bennett2023}] implies that 
$$
\frac{M(i)}{n} \le  m(i/n) + 3 n^{-1/3} \sfe^{4 i/n} = (1 + o(1)) m(i/n)
$$
for all $0 \le i \le \min\{I, \eps n\}$ with probability at least $1 - 2 \exp\left( -\frac{n^{4/3}}{2( \eps n + 2n^{2/3})}\right) \ge 1 - o(1/n)$. Since $\mb{P}[\cap_{t=0}^{\eps n} Q_{t}] \ge 1 - o(1/n)$ by \Cref{prop:de_dominance}, we know that $\min\{I, \eps n\} = \eps n$ with probability at least $1 - o(1/n)$. By taking a union bound,
this completes the proof of \Cref{prop:de_dominance}.

\bibliographystyle{amsalpha}
\bibliography{bibliography}

\end{document}